\pdfoutput=1
\documentclass[12pt]{article}
\usepackage[hmargin=1in,top=1in,bottom=1.25in]{geometry}

\usepackage{style/qstyle}
\usepackage[normalem]{ulem}
\usepackage{caption}

\renewcommand{\emph}[1]{\textbf{#1}}
\renewcommand{\mathbf}[1]{\textit{\textbf{#1}}}

\newcommand{\mbf}[1]{\mathbf{#1}}

\newcommand{\txor}{\text{TXOR}}
\newcommand{\xor}{\text{XOR}}

\newcommand{\ethreelin}{\text{E3-LIN}}
\newcommand{\suctxor}{\text{TXOR-MIP}}
\newcommand{\sucxor}{\text{XOR-MIP}}
\newcommand{\txormip}{\text{TXOR-MIP}}

\newcommand{\sucethreelin}{\text{E3-LIN-MIP}}

\newcommand{\A}{\mathcal{A}}
\newcommand{\B}{\mathcal{B}}
\newcommand{\tiltedsquare}{\hspace*{-0.14mm}\rotatebox[origin=c]{-30}{$\scriptstyle\square$}}
\newcommand{\orac}{\text{orac}}
\newcommand{\omegaco}{\omega^{\text{co}}}
\newcommand{\omegaqo}{\omega^\ast_{\text{orac}}}
\newcommand{\omegatr}{\omega^\ast_{\tr}}
\newcommand{\co}{\text{co}}

\title{\Large {\LARGE XOR Games at Full Tilt} \\ {\large The Hardness of Binary Nonlocal Games}}

\author{Richard Cleve%
\,\thanks{Institute for Quantum Computing and Cheriton School of Computer Science, University of Waterloo.}
\qquad
Eric Culf%
\,\thanks{Institute for Quantum Computing and Department of Applied Mathematics, University of Waterloo.}
\qquad
Aviv Taller%
\,\thanks{École Polytechnique Fédérale de Lausanne. Part of this work was carried out while at the Weizmann Institute of Science.}
}

\date{}

\begin{document}
	
\maketitle

\begin{abstract}
It is well known that the quantum value of an XOR nonlocal game, where the winning condition depends only on the XOR of the two players' output bits, may be approximated in polynomial time. We study a variant of the XOR game model, which we call \emph{tilted XOR games}, where the winning condition can additionally depend on only one of the output bits. We show that this dramatically increases the expressive power: the computational complexity of the problem of approximating the quantum value of tilted XOR games to constant precision is \tsf{RE}-complete. Also, our result extends to \emph{succinct} versions of tilted XOR games, where the questions can be polynomial-length binary strings,  generated by a polynomial-time verifier.

For classical strategies, the distinction between XOR games and tilted XOR games is inconsequential. H{\aa}stad (\textit{J.~ACM}, 2001) shows that they are both \tsf{NP}-complete to approximate, by using a reduction from linear systems to XOR games.
Our approach is to show that this is also quantum-sound, but as a reduction from linear system games to tilted XOR games.

Since titled XOR games are a special case of \emph{binary games} (where each party outputs a single bit), our result implies that binary games are \tsf{RE}-hard to approximate.

\end{abstract}

\newpage

\setcounter{tocdepth}{3}

\tableofcontents

\newpage

\section{Introduction}\label{sec:intro}

In a nonlocal game, two physically separated parties (or players), Alice and Bob, 
are given questions sampled from some fixed probability distribution and are required to produce answers without any communication between them. The parties are said to win the nonlocal game if they satisfy a predetermined predicate depending on both the questions and answers.

The \emph{classical value} of a nonlocal game is the maximal success probability achievable by a strategy that is restricted to using correlations permitted by classical physics; the \emph{quantum value} is the supremal success probability attainable by strategies that can utilise pre-shared entanglement, accessing the stronger correlations permitted by quantum physics. Herein, we assume the tensor product model of entanglement; the alternate commuting-operator model further increases the space of correlations. It it known that, in general, computing an approximation of the classical value of a nonlocal game is an \tsf{NP}-complete problem (where the input is not presented succinctly, so the size scales in the total number of questions)~\cite{AS98,ALMSS98}.
On the other hand, computing an approximation of the quantum value of a general nonlocal game is undecidable, in fact \tsf{RE}-complete, due to the celebrated \tsf{MIP*=RE} result~\cite{JNV+20arxiv}.

We concentrate to the restricted class of nonlocal games, called \emph{binary games}, where Alice and Bob's outputs are single bits.
Further restricting gives the \emph{XOR games}, which are binary nonlocal games with the additional property that the winning condition depends only on the players' answer bits via their XOR.

A remarkable property of XOR games is that approximating their quantum value is much easier than approximating their classical value.
Specifically, the results of~\cite{Has01} imply, for any $\varepsilon > 0$, that distinguishing between classical success probability $\ge \frac{3}{4} - \varepsilon$ and $< \frac{11}{16}+\varepsilon$ is an \tsf{NP}-complete problem.
On the other hand, XOR games have the special property that their quantum value can be formulated as a semidefinite program, and therefore approximated in polynomial time~\cite{Tsi87,CHTW04}.
In fact, the precision can be exponentially high, meaning that the approximation error can be made exponentially small in the instance size using a polynomial-time algorithm.

How much do general binary games differ from XOR games?
One property of XOR games that carries over to binary games is that there is a perfect quantum strategy (attaining success probability~$1$) if and only if there is a perfect classical strategy%
\footnote{This is generally very far from true. An example of a game that has a perfect quantum strategy but does \emph{not} have a perfect classical strategy is the Magic Square game~\cite{Mer90,Per90a}.}%
~\cite{CHTW04}.
This implies that there exists a polynomial time algorithm for determining whether or not a binary game has a perfect quantum strategy.
Also, Beigi~\cite{Bei10} shows that there exists a semidefinite programming relaxation of any binary game that approximates its value to within relative error $\approx 0.68$. On the other hand, Russell~\cite{Rus20} proves that the negative answer to Connes' embedding problem can already be witnessed at the level of two-output correlation scenarios. Combined with $\tsf{MIP}^\ast=\tsf{RE}$~\cite{JNV+20arxiv}, this yields a strict separation between approximate finite-dimensional and commuting-operator binary correlations, suggesting that binary games might already exhibit computational intractability; however, to the best of our knowledge, such an intractability statement does not follow from Russell’s result, since it does not provide a gap-preserving reduction.

\subsection{Summary of results}

As far as we know, prior to this work, the complexity of approximating the quantum value of binary games to within any constant precision was unknown.
Our results rule out the possibility that this might be easy by showing that approximating the quantum value of binary games is computationally intractable in a very strong sense: the problem is \tsf{RE}-hard, even when only a $\delta$-additive approximation is required, for some fixed constant $\delta > 0$. Therefore, the full complexity of the hardness of approximation for the quantum value of nonlocal games is captured by the binary games.
Note the sharp contrast with XOR games, that whose quantum value can be approximated to exponentially fine additive precision in polynomial time.

The binary games that we show to be \tsf{RE}-hard are found by only a slight modification of the definition of an XOR game.
In the usual definition of an XOR game, if $a$ and $b$ are Alice and Bob's respective answer bits for questions $x \in X$ and $y \in Y,$ then the winning condition is of the form 
\begin{align}
    a \oplus b = f(x,y),
\end{align}
for some function $f : X \times Y \rightarrow \{0,1\}$.
The \emph{CHSH game}~\cite{CHSH69} is a simple example of such a game, where $X = Y = \{0,1\}$ and $f(x,y) = x \wedge y$.

A variant of the CHSH game is the \emph{tilted CHSH game}~\cite{AMP12}, which can be viewed intuitively as follows: with some probability $p$, the verifier plays the CHSH game; and, with probability $1-p$, the verifier sends question $0$ to Bob and requires answer $0$ to win. There are other equivalent ways of presenting this game in the literature.
One way of describing the tilted CHSH game is with question sets $X = \{0,1,\perp\}$ and $Y = \{0,1\}$, and winning condition 
\begin{align}\label{eq:tilted-chsh}
    \begin{cases}
        a \oplus b = x\wedge y & \mbox{if $x \neq \perp$}\\
        \phantom{a \oplus \hspace*{2mm}} b = 0 & \mbox{if $x = \perp$.}
    \end{cases}
\end{align}
Also, the probability distribution on questions is a convex combination of uniform distributions on the sets $\{0,1\}\times\{0,1\}$ and $\{(\perp,0)\}$.
Note that this is a binary game%
\footnote{Without loss of generality, we can assume that Alice provides an answer bit even for the question $\perp$.}
that is \emph{not} an XOR game.

Along these lines, we define a \emph{tilted XOR game} as a binary game where there is a special question $\perp \in X$ such that the winning condition is of the form
\begin{align}
    \begin{cases}
        a \oplus b = f(x,y) & \mbox{if $x \neq \perp$}\\
        \phantom{a \oplus \hspace*{2mm}} b = f(x,y) & \mbox{if $x = \perp$.}
    \end{cases}
\end{align}
This simple tweak to the definition of an XOR game is inconsequential for classical strategies --- one may always assume that Alice answers $0$ on question $\perp$ without affecting the value.
However, we show that this modification significantly impacts the complexity of quantum strategies:
\begin{faketheorem}
There exists a constant $\delta > 0$ such that, for all $\varepsilon \in (0,\delta)$, it is \tsf{RE}-hard to distinguish between tilted XOR games having quantum value $\ge \frac{3}{4}-\varepsilon$ and those having quantum value $< \frac{3}{4}-\delta$.
\end{faketheorem}
We also consider an ungapped version of the problem, where the goal is to distinguish between quantum value $\ge \frac{3}{4}$ and $< \frac{3}{4}$, and show that this decision problem is also undecidable.

In the discussion above, we described our results in the context of \emph{explicitly specified} nonlocal games, where the game is given by an explicit description so the question sets are polynomial-size; however, we also show that they apply to \emph{succinctly specified} nonlocal games, where the verifier is specified by a polynomial-time algorithm, so the questions may be polynomial-length strings.
This version yields a two-prover \tsf{MIP*} protocol
that enables a polynomial-time verifier to determine whether or not a Turing machine halts, 
where the provers' answers are single bits and there is a constant gap between the completeness and soundness probabilities.

Our methodology builds on several previous results. In particular, a recent result of Taller and Vidick~\cite{TV25} shows that E3-LIN games --- which are the generalisation of the Magic Square game to any set of linear equations with three variables per equation --- are \tsf{RE}-hard to approximate.
Our hardness result is obtained from this by a reduction from E3-LIN games to tilted XOR games. 
The reduction that we employ is the same as that of H\aa stad~\cite{Has01} from E3-LIN games to XOR games in the classical setting.
We show that this reduction satisfies a quantum soundness property that allows it to be used to reduce E3-LIN games to tilted XOR games.
A major tool that we use and build on is an elegant characterisation of optimal and near-optimal strategies of XOR games of Slofstra~\cite{Slo11}.

The complexity of the problem of approximating the classical and quantum values of tilted XOR games varies with the completeness and soundness probabilities under consideration.
We can define the decision problem $\txor_{c,s}$ (where $\frac{1}{2}\le s \le c \le 1$) as follows.
A problem instance is a description of a tilted XOR game with the promise that its \emph{classical value} is either $\ge c$ or $< s$, and the goal is to distinguish between these two cases.
We can similarly define $\txor_{c,s}^\ast$ where the goal is to distinguish between the two cases for the \emph{quantum value}.
In \cref{fig:approx-prelim}, we summarise the known complexity categorisations pursuant to our work for $\txor_{c,s}$ and $\txor_{c,s}^\ast$ for the relevant values of $(c,s) \in [\frac{1}{2},1] \times [\frac{1}{2},1]$.
The details behind the regions in the figure are discussed in \cref{sec:hardness-for-c-and-s}.
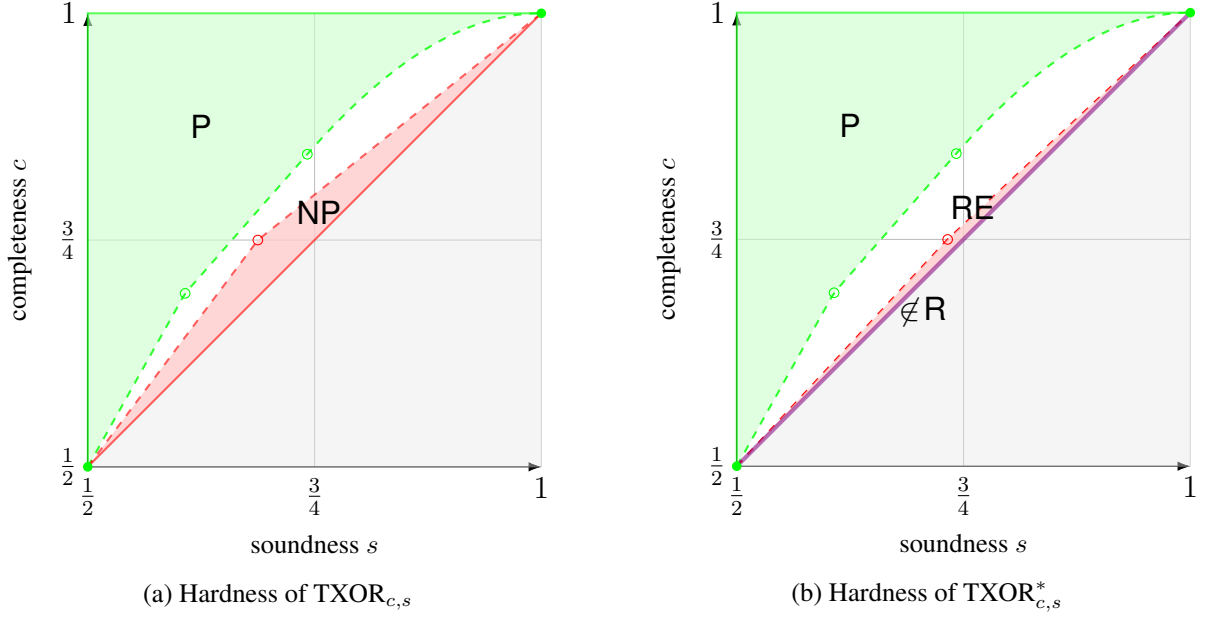
\begin{figure}[ht]
    \centering
    \begin{subfigure}{0.48\textwidth}
    \centering
    \begin{tikzpicture}[scale=0.6]
        \draw[lightgray] (0.0,0) node[below]{{\color{black}$\frac{1}{2}$}} -- (0.0,10.0) 
        (5.0,0) node[below]{{\color{black}$\frac{3}{4}$}} -- (5.0,10.0) 
        (10.0,0) node[below]{{\color{black}$1$}} -- (10.0,10.0) (0,0.0) node[left]{{\color{black}$\frac{1}{2}$}} -- (10.0,0.0) 
        (0,5.0) node[left]{{\color{black}$\frac{3}{4}$}} -- (10.0,5.0) 
        (0,10.0) node[left]{{\color{black}$1$}} -- (0.0,10.0);
        \draw[Latex-Latex] (0,10.0) -- (0,0) -- (10.0,0);
        \node at (5.0,-1.7){\footnotesize soundness $s$};
        \node at (-1.5,5.0){\rotatebox{90}{\footnotesize completeness $c$}};
        \filldraw[fill=gray!20,opacity=0.4,draw=white] (0,0) -- (10.0,0) -- (10.0,10.0) -- cycle;
        \fill[fill=green!20,opacity=0.6] (0.0,0.0) -- (0.1,0.17822000000000005) -- (0.2,0.3564400000000001) -- (0.3,0.5346600000000001) -- (0.4,0.7128800000000002) -- (0.5,0.8911000000000002) -- (0.6,1.0693200000000003) -- (0.7000000000000001,1.2475400000000003) -- (0.8,1.4257600000000004) -- (0.8999999999999999,1.6039800000000026) -- (1.0,1.7822000000000005) -- (1.1,1.9604200000000027) -- (1.2,2.1386400000000028) -- (1.3,2.3168599999999984) -- (1.4000000000000001,2.495080000000003) -- (1.5,2.6732999999999985) -- (1.6,2.8515199999999985) -- (1.7000000000000002,3.0297399999999985) -- (1.7999999999999998,3.2079599999999986) -- (1.9,3.3861799999999986) -- (2.0,3.5643999999999987) -- (2.1,3.7426199999999987) -- (2.2,3.8860400000000017) -- (2.3000000000000003,3.9998600000000017) -- (2.4,4.113680000000002) -- (2.5,4.227500000000002) -- (2.6,4.341320000000001) -- (2.7,4.455140000000002) -- (2.8000000000000003,4.568960000000002) -- (2.9,4.682780000000002) -- (3.0,4.7966000000000015) -- (3.1,4.910420000000002) -- (3.2,5.0242400000000025) -- (3.3000000000000003,5.138060000000002) -- (3.4000000000000004,5.251880000000002) -- (3.5,5.365700000000002) -- (3.5999999999999996,5.47952) -- (3.7,5.593340000000002) -- (3.8,5.70716) -- (3.9000000000000004,5.820980000000002) -- (4.0,5.9348) -- (4.1,6.04862) -- (4.2,6.16244) -- (4.3,6.276260000000001) -- (4.4,6.39008) -- (4.5,6.5039) -- (4.6000000000000005,6.61772) -- (4.699999999999999,6.731540000000001) -- (4.8,6.84536) -- (4.9,6.959127965923145) -- (5.0,7.071067811865475) -- (5.1,7.181262977631889) -- (5.2,7.289686274214113) -- (5.300000000000001,7.396310949786096) -- (5.4,7.501110696304593) -- (5.5,7.60405965600031) -- (5.6000000000000005,7.705132427757892) -- (5.699999999999999,7.804304073383294) -- (5.8,7.901550123756902) -- (5.8999999999999995,7.996846584870905) -- (6.0,8.090169943749473) -- (6.1,8.181497174250232) -- (6.2,8.270805742745617) -- (6.3,8.358073613682699) -- (6.4,8.443279255020151) -- (6.5,8.526401643540922) -- (6.6000000000000005,8.607420270039437) -- (6.7,8.686315144381911) -- (6.800000000000001,8.763066800438635) -- (6.8999999999999995,8.837656300886934) -- (7.0,8.910065241883679) -- (7.1,8.980275757606158) -- (7.199999999999999,9.048270524660191) -- (7.3,9.114032766354454) -- (7.4,9.177546256839811) -- (7.5,9.238795325112868) -- (7.6,9.29776485888251) -- (7.7,9.35444030829867) -- (7.800000000000001,9.408807689542257) -- (7.9,9.460853588275452) -- (8.0,9.510565162951536) -- (8.100000000000001,9.557930147983297) -- (8.2,9.602936856769428) -- (8.299999999999999,9.645574184577983) -- (8.4,9.685831611286309) -- (8.5,9.723699203976766) -- (8.6,9.759167619387474) -- (8.7,9.792228106217657) -- (8.8,9.822872507286887) -- (8.9,9.851093261547739) -- (9.0,9.876883405951379) -- (9.1,9.900236577165577) -- (9.200000000000001,9.921147013144777) -- (9.3,9.939609554551796) -- (9.399999999999999,9.955619646030803) -- (9.5,9.969173337331279) -- (9.6,9.980267284282718) -- (9.7,9.988898749619699) -- (9.8,9.995065603657316) -- (9.9,9.998766324816604) -- (10.0,10.0) -- (0,10.0) -- cycle;
        \draw[opacity=0.8,draw=green, thick,dashed] (0.0,0.0) -- (0.1,0.17822000000000005) -- (0.2,0.3564400000000001) -- (0.3,0.5346600000000001) -- (0.4,0.7128800000000002) -- (0.5,0.8911000000000002) -- (0.6,1.0693200000000003) -- (0.7000000000000001,1.2475400000000003) -- (0.8,1.4257600000000004) -- (0.8999999999999999,1.6039800000000026) -- (1.0,1.7822000000000005) -- (1.1,1.9604200000000027) -- (1.2,2.1386400000000028) -- (1.3,2.3168599999999984) -- (1.4000000000000001,2.495080000000003) -- (1.5,2.6732999999999985) -- (1.6,2.8515199999999985) -- (1.7000000000000002,3.0297399999999985) -- (1.7999999999999998,3.2079599999999986) -- (1.9,3.3861799999999986) -- (2.0,3.5643999999999987) -- (2.1,3.7426199999999987) -- (2.2,3.8860400000000017) -- (2.3000000000000003,3.9998600000000017) -- (2.4,4.113680000000002) -- (2.5,4.227500000000002) -- (2.6,4.341320000000001) -- (2.7,4.455140000000002) -- (2.8000000000000003,4.568960000000002) -- (2.9,4.682780000000002) -- (3.0,4.7966000000000015) -- (3.1,4.910420000000002) -- (3.2,5.0242400000000025) -- (3.3000000000000003,5.138060000000002) -- (3.4000000000000004,5.251880000000002) -- (3.5,5.365700000000002) -- (3.5999999999999996,5.47952) -- (3.7,5.593340000000002) -- (3.8,5.70716) -- (3.9000000000000004,5.820980000000002) -- (4.0,5.9348) -- (4.1,6.04862) -- (4.2,6.16244) -- (4.3,6.276260000000001) -- (4.4,6.39008) -- (4.5,6.5039) -- (4.6000000000000005,6.61772) -- (4.699999999999999,6.731540000000001) -- (4.8,6.84536) -- (4.9,6.959127965923145) -- (5.0,7.071067811865475) -- (5.1,7.181262977631889) -- (5.2,7.289686274214113) -- (5.300000000000001,7.396310949786096) -- (5.4,7.501110696304593) -- (5.5,7.60405965600031) -- (5.6000000000000005,7.705132427757892) -- (5.699999999999999,7.804304073383294) -- (5.8,7.901550123756902) -- (5.8999999999999995,7.996846584870905) -- (6.0,8.090169943749473) -- (6.1,8.181497174250232) -- (6.2,8.270805742745617) -- (6.3,8.358073613682699) -- (6.4,8.443279255020151) -- (6.5,8.526401643540922) -- (6.6000000000000005,8.607420270039437) -- (6.7,8.686315144381911) -- (6.800000000000001,8.763066800438635) -- (6.8999999999999995,8.837656300886934) -- (7.0,8.910065241883679) -- (7.1,8.980275757606158) -- (7.199999999999999,9.048270524660191) -- (7.3,9.114032766354454) -- (7.4,9.177546256839811) -- (7.5,9.238795325112868) -- (7.6,9.29776485888251) -- (7.7,9.35444030829867) -- (7.800000000000001,9.408807689542257) -- (7.9,9.460853588275452) -- (8.0,9.510565162951536) -- (8.100000000000001,9.557930147983297) -- (8.2,9.602936856769428) -- (8.299999999999999,9.645574184577983) -- (8.4,9.685831611286309) -- (8.5,9.723699203976766) -- (8.6,9.759167619387474) -- (8.7,9.792228106217657) -- (8.8,9.822872507286887) -- (8.9,9.851093261547739) -- (9.0,9.876883405951379) -- (9.1,9.900236577165577) -- (9.200000000000001,9.921147013144777) -- (9.3,9.939609554551796) -- (9.399999999999999,9.955619646030803) -- (9.5,9.969173337331279) -- (9.6,9.980267284282718) -- (9.7,9.988898749619699) -- (9.8,9.995065603657316) -- (9.9,9.998766324816604) -- (10.0,10.0);
        \draw[opacity=0.6,draw=green, thick] (0.0,0.0) -- (0,10.0) -- (10.0,10.0);
        \fill[red!20,opacity=0.8] (0.0,0.0) -- (3.75,5.0) -- (10.0,10.0) -- cycle;
        \draw[opacity=0.6,red,dashed,thick] (0.0,0.0) -- (3.75,5.0) -- (10.0,10.0);
        \draw[opacity=0.6,draw=red,thick] (0.0,0.0) -- (10.0,10.0);
        \draw[red] (3.75,5.0) circle (3pt);
        \draw[green] (2.145999999999999,3.8246011999999974) circle (3pt);
        \draw[green] (4.840400000000001,6.891343280000002) circle (3pt);
        \fill[green] (10,10) circle (3pt);
        \fill[green] (0,0) circle (3pt);
        \node at (2.5,7.5) {$\tsf{P}$};
        \node at (5.1,5.6) {$\tsf{NP}$};
    \end{tikzpicture}
\caption{Hardness of $\txor_{c,s}$}
\label{fig:classical-approx-prelim}
\end{subfigure}
\hfill
\begin{subfigure}{0.48\textwidth}
    \centering
    \begin{tikzpicture}[scale=0.6]
        \draw[lightgray] (0.0,0) node[below]{{\color{black}$\frac{1}{2}$}} -- (0.0,10.0) 
        (5.0,0) node[below]{{\color{black}$\frac{3}{4}$}} -- (5.0,10.0) 
        (10.0,0) node[below]{{\color{black}$1$}} -- (10.0,10.0) (0,0.0) node[left]{{\color{black}$\frac{1}{2}$}} -- (10.0,0.0) 
        (0,5.0) node[left]{{\color{black}$\frac{3}{4}$}} -- (10.0,5.0) 
        (0,10.0) node[left]{{\color{black}$1$}} -- (0,10.0);\draw[Latex-Latex] (0,10.0) -- (0,0) -- (10.0,0);
        \node at (5.0,-1.7){\footnotesize soundness $s$};
        \node at (-1.5,5.0){\rotatebox{90}{\footnotesize completeness $c$}};
        \filldraw[fill=gray!20,opacity=0.4,draw=white] (0,0) -- (10.0,0) -- (10.0,10.0) -- cycle;
        \fill[fill=green!20,opacity=0.6] (0.0,0.0) -- (0.1,0.17822000000000005) -- (0.2,0.3564400000000001) -- (0.3,0.5346600000000001) -- (0.4,0.7128800000000002) -- (0.5,0.8911000000000002) -- (0.6,1.0693200000000003) -- (0.7000000000000001,1.2475400000000003) -- (0.8,1.4257600000000004) -- (0.8999999999999999,1.6039800000000026) -- (1.0,1.7822000000000005) -- (1.1,1.9604200000000027) -- (1.2,2.1386400000000028) -- (1.3,2.3168599999999984) -- (1.4000000000000001,2.495080000000003) -- (1.5,2.6732999999999985) -- (1.6,2.8515199999999985) -- (1.7000000000000002,3.0297399999999985) -- (1.7999999999999998,3.2079599999999986) -- (1.9,3.3861799999999986) -- (2.0,3.5643999999999987) -- (2.1,3.7426199999999987) -- (2.2,3.8860400000000017) -- (2.3000000000000003,3.9998600000000017) -- (2.4,4.113680000000002) -- (2.5,4.227500000000002) -- (2.6,4.341320000000001) -- (2.7,4.455140000000002) -- (2.8000000000000003,4.568960000000002) -- (2.9,4.682780000000002) -- (3.0,4.7966000000000015) -- (3.1,4.910420000000002) -- (3.2,5.0242400000000025) -- (3.3000000000000003,5.138060000000002) -- (3.4000000000000004,5.251880000000002) -- (3.5,5.365700000000002) -- (3.5999999999999996,5.47952) -- (3.7,5.593340000000002) -- (3.8,5.70716) -- (3.9000000000000004,5.820980000000002) -- (4.0,5.9348) -- (4.1,6.04862) -- (4.2,6.16244) -- (4.3,6.276260000000001) -- (4.4,6.39008) -- (4.5,6.5039) -- (4.6000000000000005,6.61772) -- (4.699999999999999,6.731540000000001) -- (4.8,6.84536) -- (4.9,6.959127965923145) -- (5.0,7.071067811865475) -- (5.1,7.181262977631889) -- (5.2,7.289686274214113) -- (5.300000000000001,7.396310949786096) -- (5.4,7.501110696304593) -- (5.5,7.60405965600031) -- (5.6000000000000005,7.705132427757892) -- (5.699999999999999,7.804304073383294) -- (5.8,7.901550123756902) -- (5.8999999999999995,7.996846584870905) -- (6.0,8.090169943749473) -- (6.1,8.181497174250232) -- (6.2,8.270805742745617) -- (6.3,8.358073613682699) -- (6.4,8.443279255020151) -- (6.5,8.526401643540922) -- (6.6000000000000005,8.607420270039437) -- (6.7,8.686315144381911) -- (6.800000000000001,8.763066800438635) -- (6.8999999999999995,8.837656300886934) -- (7.0,8.910065241883679) -- (7.1,8.980275757606158) -- (7.199999999999999,9.048270524660191) -- (7.3,9.114032766354454) -- (7.4,9.177546256839811) -- (7.5,9.238795325112868) -- (7.6,9.29776485888251) -- (7.7,9.35444030829867) -- (7.800000000000001,9.408807689542257) -- (7.9,9.460853588275452) -- (8.0,9.510565162951536) -- (8.100000000000001,9.557930147983297) -- (8.2,9.602936856769428) -- (8.299999999999999,9.645574184577983) -- (8.4,9.685831611286309) -- (8.5,9.723699203976766) -- (8.6,9.759167619387474) -- (8.7,9.792228106217657) -- (8.8,9.822872507286887) -- (8.9,9.851093261547739) -- (9.0,9.876883405951379) -- (9.1,9.900236577165577) -- (9.200000000000001,9.921147013144777) -- (9.3,9.939609554551796) -- (9.399999999999999,9.955619646030803) -- (9.5,9.969173337331279) -- (9.6,9.980267284282718) -- (9.7,9.988898749619699) -- (9.8,9.995065603657316) -- (9.9,9.998766324816604) -- (10.0,10.0) -- (0,10.0) -- cycle;
        \draw[opacity=0.8,draw=green,thick,dashed] (0.0,0.0) -- (0.1,0.17822000000000005) -- (0.2,0.3564400000000001) -- (0.3,0.5346600000000001) -- (0.4,0.7128800000000002) -- (0.5,0.8911000000000002) -- (0.6,1.0693200000000003) -- (0.7000000000000001,1.2475400000000003) -- (0.8,1.4257600000000004) -- (0.8999999999999999,1.6039800000000026) -- (1.0,1.7822000000000005) -- (1.1,1.9604200000000027) -- (1.2,2.1386400000000028) -- (1.3,2.3168599999999984) -- (1.4000000000000001,2.495080000000003) -- (1.5,2.6732999999999985) -- (1.6,2.8515199999999985) -- (1.7000000000000002,3.0297399999999985) -- (1.7999999999999998,3.2079599999999986) -- (1.9,3.3861799999999986) -- (2.0,3.5643999999999987) -- (2.1,3.7426199999999987) -- (2.2,3.8860400000000017) -- (2.3000000000000003,3.9998600000000017) -- (2.4,4.113680000000002) -- (2.5,4.227500000000002) -- (2.6,4.341320000000001) -- (2.7,4.455140000000002) -- (2.8000000000000003,4.568960000000002) -- (2.9,4.682780000000002) -- (3.0,4.7966000000000015) -- (3.1,4.910420000000002) -- (3.2,5.0242400000000025) -- (3.3000000000000003,5.138060000000002) -- (3.4000000000000004,5.251880000000002) -- (3.5,5.365700000000002) -- (3.5999999999999996,5.47952) -- (3.7,5.593340000000002) -- (3.8,5.70716) -- (3.9000000000000004,5.820980000000002) -- (4.0,5.9348) -- (4.1,6.04862) -- (4.2,6.16244) -- (4.3,6.276260000000001) -- (4.4,6.39008) -- (4.5,6.5039) -- (4.6000000000000005,6.61772) -- (4.699999999999999,6.731540000000001) -- (4.8,6.84536) -- (4.9,6.959127965923145) -- (5.0,7.071067811865475) -- (5.1,7.181262977631889) -- (5.2,7.289686274214113) -- (5.300000000000001,7.396310949786096) -- (5.4,7.501110696304593) -- (5.5,7.60405965600031) -- (5.6000000000000005,7.705132427757892) -- (5.699999999999999,7.804304073383294) -- (5.8,7.901550123756902) -- (5.8999999999999995,7.996846584870905) -- (6.0,8.090169943749473) -- (6.1,8.181497174250232) -- (6.2,8.270805742745617) -- (6.3,8.358073613682699) -- (6.4,8.443279255020151) -- (6.5,8.526401643540922) -- (6.6000000000000005,8.607420270039437) -- (6.7,8.686315144381911) -- (6.800000000000001,8.763066800438635) -- (6.8999999999999995,8.837656300886934) -- (7.0,8.910065241883679) -- (7.1,8.980275757606158) -- (7.199999999999999,9.048270524660191) -- (7.3,9.114032766354454) -- (7.4,9.177546256839811) -- (7.5,9.238795325112868) -- (7.6,9.29776485888251) -- (7.7,9.35444030829867) -- (7.800000000000001,9.408807689542257) -- (7.9,9.460853588275452) -- (8.0,9.510565162951536) -- (8.100000000000001,9.557930147983297) -- (8.2,9.602936856769428) -- (8.299999999999999,9.645574184577983) -- (8.4,9.685831611286309) -- (8.5,9.723699203976766) -- (8.6,9.759167619387474) -- (8.7,9.792228106217657) -- (8.8,9.822872507286887) -- (8.9,9.851093261547739) -- (9.0,9.876883405951379) -- (9.1,9.900236577165577) -- (9.200000000000001,9.921147013144777) -- (9.3,9.939609554551796) -- (9.399999999999999,9.955619646030803) -- (9.5,9.969173337331279) -- (9.6,9.980267284282718) -- (9.7,9.988898749619699) -- (9.8,9.995065603657316) -- (9.9,9.998766324816604) -- (10.0,10.0);
        \draw[opacity=0.6,draw=green, thick] (0.0,0.0) -- (0,10.0) -- (10.0,10.0);
        \draw[red,dashed,thick] (0.0,0.0) -- (4.65,5.0) -- (10.0,10.0);
        \fill[red!20,opacity=0.8] (0.0,0.0) -- (4.65,5.0) -- (10.0,10.0) -- cycle;
        \draw[violet,opacity=0.55,ultra thick] (0.0,0.0) -- (10.0,10.0);
        \draw[red] (4.65,5.0) circle (3pt);
        \draw[green] (2.145999999999999,3.8246011999999974) circle (3pt);
        \draw[green] (4.840400000000001,6.891343280000002) circle (3pt);
        \fill[green] (10,10) circle (3pt);
        \fill[green] (0,0) circle (3pt);
        \node at (2.5,7.5) {$\tsf{P}$};
        \node at (5.2,5.7) {$\tsf{RE}$};
        \node at (4.1,3.4) {$\not\in\!\tsf{R}$};
\end{tikzpicture}
\caption{Hardness of $\txor_{c,s}^\ast$}
\label{fig:quantum-approx-prelim}
\end{subfigure}
\caption{Hardness of tilted XOR games. The \tsf{P} regions (green) are problems solvable in polynomial time, the triangular regions (red) are hard instances, and the white regions are where the complexity is unknown. The problem is undefined below the diagonal (gray). The $\tsf{RE}$-complete region is not to scale (though it is the interior of a triangle with positive area).
The diagonal line labelled $\not\in\!\mathsf{R}$ (purple) is a region where the problem is not decidable (but not known to be \tsf{RE}-hard).}\label{fig:approx-prelim}
\end{figure}

The red triangular region labelled \tsf{RE} (for \tsf{RE}-complete) in \cref{fig:quantum-approx-prelim} is our main contribution, along with the purple diagonal line labelled $\not\in\hspace*{0.3mm}$\tsf{R} (for ``not recursive'').
This red triangular region is much thinner than in the illustration: its maximum horizontal width is the constant $\delta > 0$ that we obtain in our main theorem, which evaluates to approximately $10^{-8}$.
A natural open question for further investigation is whether \tsf{RE}-hardness holds for larger values of $\delta$, and more generally to deterrmine the computational complexity within the unclassified (white) regions in~\cref{fig:approx-prelim}.

The results described above concern the standard notion of quantum strategy and a basic definition of tilted XOR games.
In  \cref{sec:other-models} {\bf \nameref{sec:other-models}} we establish results for several related settings, including alternative models of entanglement and variants of the tilted XOR game model. The modified games we consider are games with a single non-XOR type constraint (\cref{sec:single-bit-type}) and oracularised tilted XOR games (\cref{sec:oracularised}); the alternative models of entanglement we study are the commuting-operator model (\cref{sec:commuting-operator}), the oracularisable model (\cref{subsec:oracular_value_is_hard}), and the tracial model (\cref{sec:tracial-strategies}).

\newpage

One consequence of this work concerns the \emph{noncommutative max-cut problem}~\cite{CMS2024}. Given a graph $G$, the objective is to maximize
\begin{align}
    \frac{1}{|E(G)|}\sum_{(i,j) \in E(G)} \frac{1 - \tr(X_iX_j)}{2},
\end{align}
where each $X_i$ is hermitian and satisfies $X_i^2=I$. Note that $\tr$ refers to the normalised trace.
This problem is in \tsf{P} (to inverse-exponential precision), whereas its classical restriction, obtained by requiring each $X_i\in\{\pm 1\}$, is \tsf{NP}-hard to approximate within constant precision.
Our results imply that imposing the single constraint $X_1=I$ makes the problem \tsf{RE}-hard to approximate within constant precision.

\section{Preliminaries}\label{sec:preliminaries}

\subsection{Notation}

We often use bold font to denote elements of a cartesian power: $\mbf{x}=(x_1,\ldots,x_k)\in X^k$.

We consider only probability distributions $\pi$ on finite sets $X$, so we identify them with functions $\pi:X\rightarrow[0,1]$ such that $\sum_{x\in X}\pi(x)=1$.

Write $\norm{\cdot}$ for the operator norm on $\mc{B}(H)$ and $\norm{\cdot}_F$ for the Frobenius norm $\norm{T}_F=\sqrt{\Tr(T^\ast T)}$. The operator absolute value is $|T|=\sqrt{T^\ast T}$. We denote the commutator of two operators as $[S,T]=ST-TS$ and the anticommutator as $\{S,T\}=ST+TS$. For a hermitian operator~$T$, write $\mrm{sgn}(T)$ for the sign of $T$, that is $\mrm{sgn}(T)=2P-I$, where $P$ is the projector onto the nonnegative eigenspaces of $T$.

A \emph{positive operator-valued measurement (POVM)} on a Hilbert space $H$ with a finite set of outcomes $A$ is a set of positive semidefinite operators $\{P_a\}_{a\in A}\subseteq\mc{B}(H)$ such that $\sum_{a\in A}P_a=I$. A \emph{projection-valued measurement (PVM)} is a POVM such that the $P_a$ are projections.

A POVM is \emph{binary} if $A=\{0,1\}$. The \emph{observable} of a binary POVM $\{P_0,P_1\}$ is $P_0-P_1$. Every hermitian operator $A$ such that $-I\leq A\leq I$ is an observable, and it is the observable of a PVM if and only if it is unitary.

\subsection{Nonlocal games}\label{sec:nonlocal}

\begin{definition}\label{def:nonlocal}
    A \emph{nonlocal game} is a tuple $G=(X,Y,A,B,\pi,V)$ where $X,Y,A,B$ are finite sets, called Alice's questions, Bob's questions, Alice's answers, and Bob's answers, respectively; $\pi$ is a probability distribution on $X\times Y$, called the question distribution; and $V:A\times B\times X\times Y\rightarrow\{0,1\}$ is called the predicate.
    
    An operational way of viewing the execution of a nonlocal game is as follows.
    A Verifier sends Alice and Bob questions, $x \in X$ and $y \in Y$ (respectively), generated according to the distribution~$\pi$.
    Alice and Bob, who are forbidden from communicating with each other, produce answers $a \in A$ and $b \in B$ (respectively), and send them to the Verifier, who accepts the answers as a win if and only if $V(a,b|x,y) = 1$. 
\end{definition}

\begin{definition}
    A nonlocal game is \emph{binary} if $A=B=\{0,1\}$.
\end{definition}

Intuitively, a strategy for a nonlocal game is a mechanism by which Alice and Bob produce answers to their questions (without communicating with each other).
There are different types of strategies, depending on what resources of randomness and entanglement can be utilized.
Formal definitions follow.

\begin{definition}    
    A \emph{strategy} for a nonlocal game $G$ is a function $p:A\times B\times X\times Y\rightarrow[0,1]$ such that $p(\cdot,\cdot|x,y)$ is a probability distribution on $\A\times \B$ for all $x\in X$ and $y\in Y$. The \emph{value}, or winning probability, of a strategy $p$ for $G$ is
    \begin{align}
        \omega(G,p)=\sum_{\substack{x\in X,y\in Y\\a\in A,b\in B}}\pi(x,y)V(a,b|x,y)p(a,b|x,y).
    \end{align}

    A strategy $p$ is
    \begin{itemize}
        \item \emph{deterministic} if there exist functions $g:X\rightarrow A$ and $h:Y\rightarrow B$ such that $p(a,b|x,y)=\delta_{a,g(x)}\delta_{b,h(y)}$.
        \item \emph{classical} if it belongs to the convex hull of the deterministic strategies.
        \item \emph{quantum} if there exist finite-dimensional Hilbert spaces $H_A$ and $H_B$, POVMs $\{P^x_a\}_{a\in A}\subseteq\mc{B}(H_A)$ for all $x\in X$ and $\{Q^y_b\}_{b\in B}\subseteq\mc{B}(H_B)$ for all $y\in Y$, and a state $\ket{\psi}\in H_A\otimes H_B$ such that $p(a,b|x,y)=\braket{\psi}{P^{x}_a\otimes Q^y_b}{\psi}$. Using Naimark's dilation theorem, the POVMs can always chosen to be PVMs.
    \end{itemize}
    The \emph{classical value} of $G$ is the supremum over the values of all deterministic (or equivalently classical) strategies; it is denoted $\omega(G)$. The \emph{quantum value} of $G$ is the supremum over the values of all quantum strategies; it is denoted $\omega^\ast(G)$. We say a quantum strategy $p$ is \emph{optimal} if $\omega(G,p)=\omega^\ast(G)$, and \emph{$\varepsilon$-optimal} if $\omega(G,p)\geq\omega^\ast(G)-\varepsilon$. A strategy is \emph{perfect} if it has value $1$; a perfect strategy must be optimal.
\end{definition}

\cref{sec:other-models} introduces additional classes of strategies and their corresponding values for a game~$G$. Specifically, \cref{sec:commuting-operator} defines \emph{commuting operator} strategies, corresponding to $\omega^{\mathrm{co}}(G)$, \cref{subsec:oracular_value_is_hard} defines \emph{quantum oracularisable} strategies, corresponding to $\omegaqo(G)$, and \cref{sec:tracial-strategies} defines \emph{quantum tracial} strategies, corresponding to $\omegatr(G)$.

\subsection{XOR games}\label{sec:prelim-xor}

A definition that is essentially equivalent to the following%
\footnote{The original definition allowed $f$ to be a relation, in the sense that $f(x,y) \subseteq \{0,1\}$, and the winning condition is $a \oplus b \in f(x,y)$; furthermore, $f$ can be randomly generated.
Following~\cite{CSUU08}, we refer to XOR games in this more general sense as \emph{degenerate XOR games}.}
appears in~\cite{CHTW04}.

\begin{definition}[XOR game]
    An \emph{XOR game} is a binary nonlocal game $G$ for which there exists a function $f:X\times Y\rightarrow\{0,1\}$ such that $V(a,b|x,y)=\delta_{a \oplus b,f(x,y)}$.
\end{definition}

The value of a strategy $p$ for an XOR game $G$ has a simple expression in terms of the observables $A_x$ and $B_y$ for the POVMs $\{P^x_a\}_a$ and $\{Q^y_b\}_b$, respectively:
\begin{align*}
    \omega(G,p)&=\sum_{\substack{x\in X,y\in Y\\a,b\in\{0,1\}:\,a\oplus b=f(x,y)}}\pi(x,y)\braket{\psi}{P^x_{a}\otimes Q^y_b}{\psi}\\
    &=\frac{1}{4}\sum_{x\in X,y\in Y}\pi(x,y)\sum_{a\in\{0,1\}}\braket{\psi}{(I+(-1)^{a}A_x)\otimes(I+(-1)^{a+f(x,y)}B_y)}{\psi}\\
    &=\frac{1}{2}+\frac{1}{2}\sum_{x\in X,y\in Y}(-1)^{f(x,y)}\pi(x,y)\braket{\psi}{A_x\otimes B_y}{\psi}.
\end{align*}
This presentation implies that the XOR game can be fully characterised by the matrix $H$ with entries $H_{x,y}=(-1)^{f(x,y)}\pi(x,y)$ and the strategy can be characterised by the \emph{quantum correlation} $C(x,y)=\braket{\psi}{A_x\otimes B_y}{\psi}$. In general, a \emph{correlation} is a function $C:X\times Y\rightarrow[-1,1]$ and the \emph{bias} of a correlation is
\begin{align*}
    \beta(G,C)=\sum_{x\in X,y\in Y}H_{x,y}C(x,y).
\end{align*}
For a quantum strategy $p$ with associated correlation $C$, $\beta(G,C)=2\omega(G,p)-1$. The \emph{quantum bias} of $G$ is $\beta^\ast(G)=2\omega^\ast(G)-1$.

\begin{definition}
    A \emph{vector correlation} $C$ for an XOR game $G$ is a function $C:X\times Y\rightarrow[-1,1]$ such that $C(x,y)=\braket{u_x}{v_y}$ for some real unit vectors $\ket{u_x},\ket{v_y}$.

    The \emph{vector bias} is the supremum over the biases of all vector correlation, denoted $\beta^{vect}(G)$. A vector correlation $C$ is \emph{optimal} if $\beta(G,C)=\beta^{vect}(G)$, and \emph{$\varepsilon$-optimal} if $\beta(G,C)\geq\beta^{vect}(G)-2\varepsilon$.
\end{definition}

We include a factor of $2$ in the definition of $\varepsilon$-optimal bias to ensure that the value is within $\varepsilon$ of the optimal value.

The vector correlations are in bijective correspondence with quantum correlations, and therefore $\beta^{vect}(G)=\beta^\ast(G)$~\cite{Tsi87}. This characterisation via vector correlations allows the bias of an XOR game to be expressed as a semidefinite program (SDP). Following~\cite{CSUU08,Slo11}, let $B=\frac{1}{2}\squ*{\begin{smallmatrix}0&H\\H^{T}&0\end{smallmatrix}}$. Then, the quantum bias of $G$ is the optimal value of the SDP over $(|X|+|Y|)\times(|X|+|Y|)$-matrices~$M$
\begin{align}
    \begin{split}\label{eq:primal}
        \text{maximise }& \Tr(BM)\\
        \text{subject to }& M_{i,i} = 1,\;\forall\,i\in X\sqcup Y\\
        &M\geq 0.
    \end{split}
\end{align}
Here, $M$ represents the Gram matrix of the vectors in the quantum strategy. The dual of this SDP is the SDP over vectors $\mbf{a}\in\R^X$ and $\mbf{b}\in\R^Y$
\begin{align}
    \begin{split}\label{eq:dual}
        \text{minimise }&\frac{1}{2}\sum_{x\in X}a_x+\frac{1}{2}\sum_{y\in Y}b_y\\
        \text{subject to }&\frac{1}{2}\Delta(\mbf{a},\mbf{b})\geq B,
    \end{split}
\end{align}
where $\Delta(\mbf{a},\mbf{b})$ is the matrix with the entries of the vectors along the diagonal, and zeroes elsewhere. For an optimal strategy, the values $a_x$ and $b_y$ are called the \emph{optimal row and column biases}, respectively, and satisfy
\begin{align*}
    &\sum_{y\in Y}H_{x,y}\braket{u_x}{v_y}=a_x\\
    &\sum_{x\in X}H_{x,y}\braket{u_x}{v_y}=b_y.
\end{align*}

We make use of the following structure theorem for near-optimal XOR game strategies.

\begin{theorem}[\cite{Slo11} Theorem 3.1]\label{thm:xor-rigidity}
    Let $G$ be an XOR game, and let $a_x$ for $x\in X$ and $b_y$ for $y\in Y$ be the optimal row and column biases, respectively. Then, for any $\varepsilon$-optimal vector correlation $C(x,y)=\braket{u_x}{v_y}$, 
    \begin{align*}
        &\norm[\Big]{\sum_{y\in Y}H_{x,y}\ket{v_y}-a_x\ket{u_x}}^2=10\sqrt{2(|X|+|Y|)\varepsilon},\\
        &\norm[\Big]{\sum_{x\in X}H_{x,y}\ket{u_x}-b_y\ket{v_y}}^2=10\sqrt{2(|X|+|Y|)\varepsilon}.
    \end{align*}
\end{theorem}

In the special case where $\varepsilon = 0$, this simplifies to the following (which is useful in~\cref{lem:optimal-case}).

\begin{corollary}[\cite{Slo11} Corollary 3.2]\label{thm:xor-rigidity-exact}
    Let $G$ be an XOR game, and let $a_x$ for $x\in X$ and $b_y$ for $y\in Y$ be the optimal row and column biases, respectively. Then, for any optimal vector correlation $C(x,y)=\braket{u_x}{v_y}$, 
    \begin{align*}
        \sum_{y\in Y}H_{x,y}\ket{v_y} &= a_x\ket{u_x},\\
        \sum_{x\in X}H_{x,y}\ket{u_x} &= b_y\ket{v_y}.
    \end{align*}
\end{corollary}

We next review known facts about the computational complexity of determining the classical and quantum values of XOR games.
First, note that any XOR game can be trivially won with probability at least $\frac{1}{2}$ by one party simply outputting a random bit (independent of their input value).
So the range of interesting success probabilities is $(\frac{1}{2},1]$.

\begin{definition}[$\xor_{c,s}$ and $\xor_{c,s}^\ast$]
    Let $\frac{1}{2} \le s \le c \le 1$.
    Define $\xor_{c,s}$ as the problem of deciding, for a given XOR game $G$, if $\omega(G)\geq c$ or $\omega(G)<s$, with the promise that one of these two holds.
    Similarly, define $\xor_{c,s}^\ast$ as the problem of deciding, for a given XOR game $G$, if $\omega^\ast(G)\geq c$ or $\omega^\ast(G)<s$, with the promise that one of these two holds.
\end{definition}

First, consider the classical case of $\xor_{c,s}$.
Determining whether an XOR game has a perfect classical strategy reduces to the problem of solving a system of linear equations modulo 2, which is polynomial time computable.
Therefore, $\xor_{1,s} \in \tsf{P}$ for all $s$.
On the other hand, there exist $c > s$ such that $\xor_{c,s}$ is $\tsf{NP}$-complete~\cite{Has01}.
In particular, $\tsf{NP}$-completeness holds for $c=\frac{3}{4}-\varepsilon$ and $s=\frac{11}{16}+\varepsilon$ for any $\varepsilon>0$ (see \cref{sec:cube} for more details).

Next, consider the quantum case of $\xor^\ast_{c,s}$.
Since an XOR game has a perfect quantum strategy if and only if it has a perfect classical strategy, it follows that $\xor_{1,s}^\ast\in \tsf{P}$ for all $s$.
However, unlike the classical case, for \emph{any} constants $c>s$, $\xor^\ast_{c,s}\in\tsf{P}$~\cite{Tsi87,CHTW04}.
Also, for the ungapped intermediate cases where $\frac{1}{2} < c < 1$, it holds that $\xor^\ast_{c,c}\in\tsf{PSPACE}$, since this problem corresponds to the exact feasibility of the semidefinite program of Eq.~\eqref{eq:primal}. This feasibility problem is a special case of the existential theory of the reals, which is known to be decidable in $\tsf{PSPACE}$~\cite{Can88}.

We also consider the succinctly-presented versions of these problems, denoted as $\sucxor_{c,s}$ and $\sucxor_{c,s}^\ast$. Here, the problems are identical to the aforementioned, but the XOR games are not presented explicitly, but rather as (probabilistic) Turing machines that can sample the question distribution and compute the predicate.
This allows for games that have exponentially-many questions in the description size. Here, the complexities scale up accordingly, but preserve the same relationships: for all $\frac{1}{2} < s < c < 1$, it holds that $\sucxor_{1,s},\sucxor_{c,s}^\ast\in\tsf{EXP}$; for all $\varepsilon>0$, $\sucxor_{\frac{3}{4}-\varepsilon,\frac{11}{16}+\varepsilon}$ is $\tsf{NEXP}$-complete; and, for all $\frac{1}{2} < c < 1$, $\sucxor_{c,c}^\ast\in\tsf{EXPSPACE}$. 
However, the succinct picture is more nuanced in the entangled setting: the known containment $\sucxor_{c,s}^\ast\in\tsf{QIP}(2)$\cite{Wehner2006},
together with $\tsf{QIP}(2)\subseteq\tsf{PSPACE}$ \cite{JJUW11},
yields a $\tsf{PSPACE}$ upper bound for $\sucxor_{c,s}^\ast$.

\subsection{Linear constraint system (E3-LIN) games}

Informally, a 3-linear constraint system consists of a set of binary variables, and a set of equations involving mod 2 sums of triples of these variables.
A well-known example is the Magic Square game~\cite{Mer90,Per90a}, which can be expressed as these six equations in nine variables:
\begin{align*}
    v_1 \oplus v_2 \oplus v_3 &= 0 \\
    v_4 \oplus v_5 \oplus v_6 &= 0 \\
    v_7 \oplus v_8 \oplus v_9 &= 0 \\
    v_1 \oplus v_4 \oplus v_7 &= 1 \\
    v_2 \oplus v_5 \oplus v_8 &= 1 \\
    v_3 \oplus v_6 \oplus v_9 &= 1.
\end{align*}
The nonlocal game associated with any such system of equations is to send Alice an equation and Bob a variable in that equation, and then to demand that Alice returns a satisfying assignment to the equation and Bob returns a value of the variable that is consistent with Alice's returned value~\cite{CM14a}.
The following is a formal definition of the general case of such games.

\begin{definition}
    A \emph{3-linear constraint system (E3-LIN)} consists of a set $Y = \{1, 2, \dots, m\}$ (corresponding to variables $v_1,v_2,\dots,v_m$), and two sets  $R_0,R_1\subseteq \{(i,j,k) : 1 \le i < j < k \le m\}$, that correspond to the (modulo 2) linear equations of the form $v_i \oplus v_j \oplus v_k = 0$ and $v_i \oplus v_j \oplus v_k = 1$, respectively.

    Let $S=(Y,R_0,R_1)$ be a 3-LCS and let 
    \begin{align*}
        X = (\{0\}\times R_0)\cup(\{1\}\times R_1), 
    \end{align*}
    and $\pi$ be a probability distribution on $X$.
    Then, the \emph{E3-LIN game} of $(S,\pi)$ is the nonlocal game
    \begin{align*}
        G(S,\pi)=\bigl(X,Y,\{0,1\}^3,\{0,1\},\tilde{\pi},V_{S}\bigr), 
    \end{align*}
    where, for $(c, (i,j,k)) \in X$ and $\ell \in Y$,
    \begin{align*}
        &\tilde{\pi}\bigl((c, (i,j,k)),\ell\bigr)=
        \begin{cases}
        \frac{1}{3}\,\pi(c, (i,j,k)) & \text{if $\ell \in \{i,j,k\}$} \\
        \ 0 & \text{otherwise,}
        \end{cases}
    \end{align*}
    and
    \begin{align*}
        &V_S\bigl((a_i,a_j,a_k),b\big|(c, (i,j,k)),\ell\bigr)=
        \begin{cases}
        1 & \text{if $a_i\oplus a_j\oplus a_k=c$ and $b = a_{\ell}$} \\
        0 & \text{otherwise.}
        \end{cases}
    \end{align*}
\end{definition}

Perfect quantum strategies for E3-LIN games admit a useful characterisation~\cite{CM14a}. Suppose $G(S,\pi)$ has a perfect quantum strategy (and assume without loss of generality that $\pi$ is never~$0$).
Then there exist binary observables $B_1,B_2,\dots,B_m \in \mathcal{B}(\mathbb{C}^d)$ such that, for any $(c, (i,j,k)) \in X$, it holds that $B_i, B_j, B_k$ are mutually commuting and  
\begin{align}
    &B_{i}B_{j}B_{k}=(-1)^cI.
\end{align}
A perfect quantum strategy can be based on a maximally entangled state $\ket{\psi} = \frac{1}{\sqrt{d}}\sum_{r=1}^{d} \ket{r}\otimes \ket{r}$, where Alice's answer bits are the outcomes of the measurements associated with $B_i, B_j, B_k$ and Bob's answer bit is the outcome of the measurement associated with $B_{\ell}^T$.

We next review known facts about the computational complexity of determining the classical and quantum values of E3-LIN games.
First note that, for any E3-LIN, at least half of the equations can be simultaneously satisfied.
From such an assignment, there is a simple classical strategy that succeeds with probability at least $\frac{5}{6}$.
Namely, Alice outputs the value of the assignment for all satisfied equations, and the assignment with one bit flipped for all unsatisfied equations; Bob outputs the value of the assignment.
So the range of interesting success probabilities is $(\frac{5}{6},1]$.

\begin{definition}[$\ethreelin_{c,s}$ and $\ethreelin_{c,s}^\ast$]
    Let $\frac{5}{6} < s \le c \le 1$.
    Define $\ethreelin_{c,s}$ as the problem of deciding, for a given E3-LIN game $G$, if $\omega(G)\geq c$ or $\omega(G)<s$, with the promise that one of these two holds.
    Similarly, define $\ethreelin_{c,s}^\ast$ as the problem of deciding if $\omega^\ast(G)\geq c$ or $\omega^\ast(G)<s$.
    
    We call the succinctly-presented versions of these problems $\sucethreelin_{c,s}$ and $\sucethreelin_{c,s}^\ast$.
\end{definition}

The classical case is similar to XOR games in that: $\ethreelin_{1,s}\in\tsf{P}$ for all $s$; and  there exist $c>s$ such that $\ethreelin_{c,s}$ is $\tsf{NP}$-complete.  Furthermore, unlike in the case of XOR games, $\tsf{NP}$-hardness holds for $c$ arbitrarily close to $1$: we may take $c=1-\varepsilon$ and $s=\frac{5}{6}+\varepsilon$ for any $\varepsilon>0$~\cite{Has01}.

The complexities of the succinct versions of these problems are the same, but scaled up: $\sucethreelin_{1,s}\in\tsf{EXP}$ for all $s$; and, for all $\varepsilon>0$, $\sucethreelin_{1-\varepsilon,\frac{5}{6}+\varepsilon}$ is $\tsf{NEXP}$-complete.

However, in the quantum case, the complexity increases dramatically: $\text{E3-LIN}^\ast_{1,1}$ is undecidable~\cite{Slo19}; and there exist $c>s$ such that $\text{E3-LIN}^\ast_{c,s}$ is $\tsf{RE}$-complete~\cite{TV25}, with respect to computable reductions.
In the latter case,~$c$ may be arbitrarily close to $1$: we may choose $c=1-\varepsilon$ and $s=\frac{119}{120}+\varepsilon$ for any $\varepsilon > 0$. The complexity of $\text{E3-LIN}^\ast_{1,s}$ for $s<1$ is currently not well-understood~\cite{PS25}.

Similarly, for the succinct versions of these problems: $\text{E3-LIN-MIP}^\ast_{1,1}$ is undecidable~\cite{Slo19}; and there exist $c>s$ such that $\text{E3-LIN-MIP}^\ast_{c,s}$ is $\tsf{RE}$-complete~\cite{TV25}, in the sense that there exists a polynomial-time reduction from the halting problem to $\sucethreelin_{c,s}^\ast$.

\subsection{H{\aa}stad's reduction from E3-LIN games to XOR games}\label{sec:cube}

The NP-hardness of approximating the classical value of XOR games is an immediate consequence of a hardness result of H{\aa}stad~\cite{Has01}, which is proved via a gadget reduction from E3-LIN games.
On page~828 of~\cite{Has01}, the gadget is credited to Sorkin and to techniques in~\cite{TSSW00}.
We can express this gadget as an XOR game, where Alice's questions are $\{000,011,101,110\}$ and Bob's questions are 
$\{001,010,100,111\}$.
The questions can be visualized as the vertices of the cube graph in~\cref{fig:cube-game} (gray for Alice's questions and white for Bob's).
\begin{figure}[h!]
\centering
\scalebox{0.75}{
\begin{tikzpicture}[
shape_style/.style={draw=black!90!black, fill=white!10, thick, minimum size=3mm},
shape_style-g/.style={draw=black!90!black, fill=gray!15, thick, minimum size=3mm},
label_style/.style={font=\small\sffamily, text=black}
]
\draw[double,thick,darkred] (0,0) -- (3,0);
\draw[double,thick,darkred] (0,0) -- (0,3);
\draw[double,thick,darkred] (0,3) -- (3,3);
\draw[double,thick,darkred] (2.0,0.6) -- (5.0,0.6);
\draw[double,thick,darkred] (5.0,0.6) -- (5.0,3.6);
\draw[double,thick,darkred] (2.0,0.6) -- (2.0,3.6);
\draw[double,thick,darkred] (2.0,3.6) -- (5.0,3.6);
\draw[double,thick,darkred] (0,0) -- (2.0,0.6);
\draw[double,thick,darkred] (3,0) -- (5.0,0.6);
\draw[double,thick,darkred] (0,3) -- (2.0,3.6);
\draw[double,thick,darkred] (3,3) -- (5.0,3.6);
\draw[thick,blue] (2.0,0.6) -- (3.0,3.0);
\draw[line width=1.0mm,white] (3,0) -- (2.55,1.62);
\draw[thick,blue] (3,0) -- (2.0,3.6);
\draw[line width=1.3mm,white] (0,3) -- (2.25,1.92);
\draw[thick,blue] (0,3) -- (5.0,0.6);
\draw[line width=1.3mm,white] (0,0) -- (2.25,1.62);
\draw[thick,blue] (0,0) -- (5.0,3.6);
\draw[line width=1.7mm,white] (0,3) -- (3,3);
\draw[double,thick,darkred] (0,3) -- (3,3);
\draw[line width=1.6mm,white] (3,0) -- (3,3);
\draw[double,thick,darkred] (3,0) -- (3,3);
\node[shape_style-g, circle, minimum size=3mm, inner sep=2pt] (circ) at (0,0) {\small 000};
\node[shape_style, circle, minimum size=3mm, inner sep=2pt] (circ) at (3,0) {\small 100};
\node[shape_style, circle, minimum size=3mm, inner sep=2pt] (circ) at (0,3) {\small 001};
\node[shape_style-g, circle, minimum size=3mm, inner sep=2pt] (circ) at (3,3) {\small 101};
\node[shape_style, circle, minimum size=3mm, inner sep=1.6pt] (circ) at (2.0,0.6) {\small 010};
\node[shape_style-g, circle, minimum size=3mm, inner sep=1.6pt] (circ) at (5.0,0.6) {\small 110};
\node[shape_style-g, circle, minimum size=3mm, inner sep=1.6pt] (circ) at (2.0,3.6) {\small 011};
\node[shape_style, circle, minimum size=3mm, inner sep=1.6pt] (circ) at (5.0,3.6) {\small 111};
\end{tikzpicture}
}
\caption{The cube game depicted as a bipartite graph. The question pairs are the edges.
The winning condition is that the XOR of the answer bits is: 1 for double edges (red); and 0 for single edges (blue).}
\label{fig:cube-game}
\end{figure}
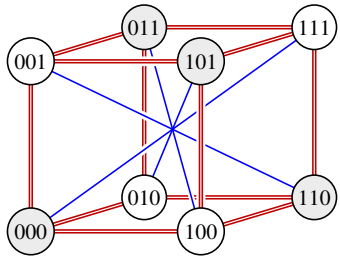

The winning condition for a pair of questions is that the XOR of the answer bits is: 1 for vertices connected by a double edge; and 0 for vertices connected by a single edge.
Below is a formal definition of this cube game.

\begin{definition}[cube game $G^{\Box}$]\label{def:cube}
Let $X_0,X_1\subseteq\{0,1\}^3$ be the sets of bit strings of even and odd parity, respectively. Then, the \emph{cube game}, that we denote as $G^{\Box}$, is the XOR game with question sets $X=X_0$ and $Y=X_1$, and game matrix 
    \begin{align}\label{eq:cube-game}
        H_{x,y}=
        \begin{cases}
            +\frac{1}{16} & \text{if $(x_1,x_2,x_3) = (\neg y_1, \neg y_2, \neg y_3)$} \\
            -\frac{1}{16} & \text{otherwise,}
        \end{cases}
    \end{align}
    for all $x \in X$ and $y \in Y$.
\end{definition}

\begin{lemma}[\cite{Has01}]\label{lem:Hastad}
    The classical value of the cube game is $\omega(G^{\Box})=\frac{3}{4}$.
    Moreover, among all deterministic strategies $c(a,b|x,y)=\delta_{a,g(x)}\delta_{b,h(y)}$ for which $g(000) \oplus g(011) \oplus g(101) \oplus g(110) = 0$, success probability $\frac{3}{4}$ is attainable; whereas, if $g(000) \oplus g(011) \oplus g(101) \oplus g(110) = 1$ then the highest success probability attainable is $\frac{5}{8}$.
\end{lemma}

The gap between the two success probabilities in~\cref{lem:Hastad} enables us to use the cube game as a gadget to test parities of bits in classical strategies.
We sketch the reduction based on this from E3-LIN games to XOR games due to \cite{Has01} and its analysis for classical strategies.

\subsubsection*{Reduction from E3-LIN games to XOR games}

Let $G(S,\pi)$ be an instance of E3-LIN. We construct an XOR game $\widetilde{G}$ that consists of several cube games (one for each equation) that intersect at some of their vertices.
For each equation of the form $v_i \oplus v_j \oplus v_k = 0$, we create a copy of the cube game, labeling vertex $000$ with the symbol $\perp$ and the other even parity vertices with $i$, $j$, and $k$.
\begin{figure}[ht!]
\centering
\begin{subfigure}{0.48\textwidth}
\centering
\scalebox{0.75}{
\begin{tikzpicture}[
shape_style/.style={draw=black!90!black, fill=white!10, thick, minimum size=2cm},
shape_style-g/.style={draw=black!90!black, fill=gray!15, thick, minimum size=3mm},
label_style/.style={font=\small\sffamily, text=green}
]
\draw[double,thick,darkred] (0,0) -- (3,0);
\draw[double,thick,darkred] (0,0) -- (0,3);
\draw[double,thick,darkred] (0,3) -- (3,3);
\draw[double,thick,darkred] (2.0,0.6) -- (5.0,0.6);
\draw[double,thick,darkred] (5.0,0.6) -- (5.0,3.6);
\draw[double,thick,darkred] (2.0,0.6) -- (2.0,3.6);
\draw[double,thick,darkred] (2.0,3.6) -- (5.0,3.6);
\draw[double,thick,darkred] (0,0) -- (2.0,0.6);
\draw[double,thick,darkred] (3,0) -- (5.0,0.6);
\draw[double,thick,darkred] (0,3) -- (2.0,3.6);
\draw[double,thick,darkred] (3,3) -- (5.0,3.6);
\draw[thick,blue] (2.0,0.6) -- (3.0,3.0);
\draw[line width=1.0mm,white] (3,0) -- (2.55,1.62);
\draw[thick,blue] (3,0) -- (2.0,3.6);
\draw[line width=1.3mm,white] (0,3) -- (2.25,1.92);
\draw[thick,blue] (0,3) -- (5.0,0.6);
\draw[line width=1.3mm,white] (0,0) -- (2.25,1.62);
\draw[thick,blue] (0,0) -- (5.0,3.6);
\draw[line width=1.7mm,white] (0,3) -- (3,3);
\draw[double,thick,darkred] (0,3) -- (3,3);
\draw[line width=1.6mm,white] (3,0) -- (3,3);
\draw[double,thick,darkred] (3,0) -- (3,3);
\node[shape_style-g, circle, minimum size=3mm, inner sep=0.3pt] (circ) at (0,0) {\small $\phantom{A_{1...}}$};\node at (0,0.03) {\large $\perp$};
\node[shape_style, circle, minimum size=3mm, inner sep=1.9pt] (circ) at (3,0) {\small $100$};
\node[shape_style, circle, minimum size=3mm, inner sep=1.9pt] (circ) at (0,3) {\small $001$};
\node[shape_style-g, circle, minimum size=3mm, inner sep=0.0pt] (circ) at (3,3) {\small $\phantom{A_{101}}$};\node at (3,3) {\large $j$};
\node[shape_style, circle, minimum size=3mm, inner sep=1.6pt] (circ) at (2.0,0.6) {\small $010$};
\node[shape_style-g, circle, minimum size=3mm, inner sep=0.0pt] (circ) at (5.0,0.6) {\small $\phantom{A_{1...}}$};\node at (5.0,0.6) {\large $k$};
\node[shape_style-g, circle, minimum size=3mm, inner sep=0.0pt] (circ) at (2.0,3.6) {\small $\phantom{A_{1...}}$};\node at (2.0,3.6) {\large $i$};
\node[shape_style, circle, minimum size=3mm, inner sep=1.6pt] (circ) at (5.0,3.6) {\small $111$};
\end{tikzpicture}
}
\caption{$v_i \oplus v_j \oplus v_k = 0$}
\end{subfigure}
\begin{subfigure}{0.48\textwidth}
\centering
\scalebox{0.75}{
\begin{tikzpicture}[
shape_style/.style={draw=black!90!black, fill=white!10, thick, minimum size=3cm},
shape_style-g/.style={draw=black!90!black, fill=gray!15, thick, minimum size=3mm},
label_style/.style={font=\small\sffamily, text=green}
]
\draw[thick,blue] (0,0) -- (3,0);
\draw[thick,blue] (0,0) -- (0,3);
\draw[double,thick,darkred] (0,3) -- (3,3);
\draw[double,thick,darkred] (2.0,0.6) -- (5.0,0.6);
\draw[double,thick,darkred] (5.0,0.6) -- (5.0,3.6);
\draw[double,thick,darkred] (2.0,0.6) -- (2.0,3.6);
\draw[double,thick,darkred] (2.0,3.6) -- (5.0,3.6);
\draw[thick,blue] (0,0) -- (2.0,0.6);
\draw[double,thick,darkred] (3,0) -- (5.0,0.6);
\draw[double,thick,darkred] (0,3) -- (2.0,3.6);
\draw[double,thick,darkred] (3,3) -- (5.0,3.6);
\draw[line width=1.8mm,white] (0,0) -- (2.25,1.62);
\draw[double,thick,darkred] (0,0) -- (5.0,3.6);
\draw[thick,blue] (2.0,0.6) -- (3.0,3.0);
\draw[line width=1.0mm,white] (3,0) -- (2.55,1.62);
\draw[thick,blue] (3,0) -- (2.0,3.6);
\draw[line width=1.2mm,white] (0,3) -- (2.25,1.92);
\draw[thick,blue] (0,3) -- (5.0,0.6);
\draw[line width=1.7mm,white] (0,3) -- (3,3);
\draw[double,thick,darkred] (0,3) -- (3,3);
\draw[line width=1.6mm,white] (3,0) -- (3,3);
\draw[double,thick,darkred] (3,0) -- (3,3);
\node[shape_style-g, circle, minimum size=3mm, inner sep=0.3pt] (circ) at (0,0) {\small $\phantom{A_{1...}}$};\node at (0,0.03) {\large $\perp$};
\node[shape_style, circle, minimum size=3mm, inner sep=1.9pt] (circ) at (3,0) {\small $100$};
\node[shape_style, circle, minimum size=3mm, inner sep=1.9pt] (circ) at (0,3) {\small $001$};
\node[shape_style-g, circle, minimum size=3mm, inner sep=0.0pt] (circ) at (3,3) {\small $\phantom{A_{101}}$};\node at (3,3) {\large $j$};
\node[shape_style, circle, minimum size=3mm, inner sep=1.6pt] (circ) at (2.0,0.6) {\small $010$};
\node[shape_style-g, circle, minimum size=3mm, inner sep=0.0pt] (circ) at (5.0,0.6) {\small $\phantom{A_{1...}}$};\node at (5.0,0.6) {\large $k$};
\node[shape_style-g, circle, minimum size=3mm, inner sep=0.0pt] (circ) at (2.0,3.6) {\small $\phantom{A_{1...}}$};\node at (2.0,3.6) {\large $i$};
\node[shape_style, circle, minimum size=3mm, inner sep=1.6pt] (circ) at (5.0,3.6) {\small $111$};
\end{tikzpicture}
}
\caption{$v_i \oplus v_j \oplus v_k = 1$}
\end{subfigure}
\caption{Gadgets for E3-LIN equations of the form $v_i \oplus v_j \oplus v_k = b$.}
\label{fig:tilted-cube-game-reduction}
\end{figure}
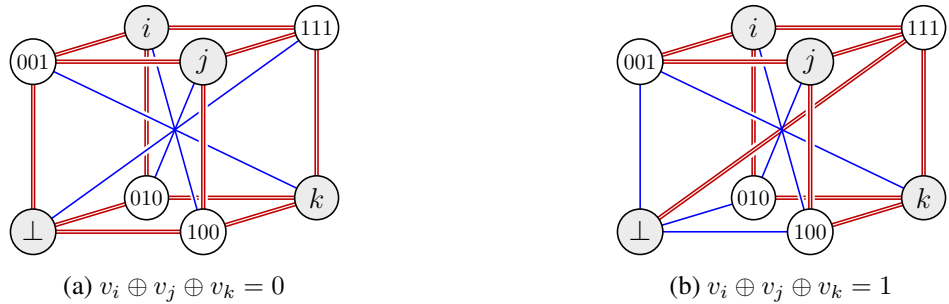
We similarly replace each equation of the form $v_i \oplus v_j \oplus v_k = 1$, but flip the double/single edges involving label $\perp$; for deterministic strategies, this is equivalent to flipping Alice's answer bit for question $\perp$ (in terms of the game matrix, in~\cref{eq:cube-game}, the sign of the entries involving $000$ is flipped). These two cubes are illustrated in \cref{fig:tilted-cube-game-reduction}.

For any equations sharing a variable, the cubes intersect at the vertex corresponding to that variable.
Also, all the cubes intersect at vertex $\perp$.
There are no intersections among the odd parity vertices; each odd parity vertex is assigned a unique label.

The resulting XOR game $\widetilde{G} = (X, Y, \{0,1\}, \{0,1\}, V)$ has $X = \{\perp,1,2,\dots,m\}$ and $Y$ consisting of all the distinct labels of the odd parity vertices.
The questions are generated randomly as follows.
First a cube is randomly selected according to the distribution associated with the E3-LIN equations.
Then that cube game is played.
Namely, Alice and Bob are queried the labels of the even and odd parity vertices (respectively) of the selected cube.

\begin{theorem}[\cite{Has01}]
    Let $G$ be an E3-LIN game and $\widetilde{G}$ the XOR game resulting from the above reduction.
    Then $\omega(G) \ge 1-\varepsilon$ implies $\omega(\widetilde{G}) \ge 1 - \frac{3}{8}\varepsilon$ and 
    $\omega(G) < \frac{5}{6} + \varepsilon$ implies $\omega(\widetilde{G}) < \frac{11}{16}+\frac{3}{8}\varepsilon$.
\end{theorem}

\begin{proof}
We may restrict our attention to deterministic strategies. Since any such strategy for an XOR game is invariant under the negation of all the answer bits, we may assume that Alice's answer bit is $0$ for question $\perp$.
Under this assumption, each cube game for an equation tests the parity of the values associated with the nodes $011$, $101$, and $011$ in the sense that, if the equation is satisfied then the value of that cube game can be $\frac{3}{4}$; otherwise the value can be at most $\frac{5}{8}$ (by~\cref{lem:Hastad}).

If $\omega(G) \ge 1 - \varepsilon$ then there exists an assignment that satisfies a fraction $\ge 1-3\varepsilon$ of the equations (weighted by their probabilities).
For each satisfied equation, the associated cube game can be won with probability $\frac{3}{4}$ and,
for each unsatisfied equation, it can be won with probability $\frac{5}{8}$.
Averaging and applying \cref{lem:Hastad}, we obtain $\omega(\widetilde{G}) \ge (1-3\varepsilon)\frac{3}{4}+(3\varepsilon)\frac{5}{8} = 1 - \frac{3}{8}\varepsilon$.

If $\omega(G) < \frac{5}{6} + \varepsilon = 1 - (\frac{1}{6}-\varepsilon)$ then the maximum fraction of satisfied equations (weighted by their probabilities) is $< 1 - 3(\frac{1}{6} - \varepsilon) = \frac{1}{2} + 3\varepsilon$.
Averaging and applying~\cref{lem:Hastad} again, we obtain $\omega(\widetilde{G}) < (\frac{1}{2} + 3\varepsilon)\frac{3}{4} + (\frac{1}{2} - 3\varepsilon)\frac{5}{8} = \frac{11}{16} + \frac{3}{8}\varepsilon$.
\end{proof}

\section{Tilted XOR games}\label{sec:tilted-xor-games}

\subsection{Definitions and basic results}

In this section, we introduce the class of binary games for which we show hardness of approximation.
The games can be understood as binary games where there are two types of questions:
\begin{itemize}
    \item
    \emph{XOR type} questions, where the winning condition is of the form $a \oplus b = f_0(x,y)$.
    \item
    \emph{bit type} questions (for Bob's bit only), where the winning condition is $b = f_1(y)$. 
\end{itemize}
The two functions $f_0 : X \times Y \rightarrow \{0,1\}$ and $f_1 : Y \rightarrow \{0,1\}$ can be amalgamated into one by adding a special symbol $\perp$ to $X$ and defining $f : X \times Y \rightarrow \{0,1\}$ as: $f(x,y) = f_0(x,y)$ if $x \neq \perp$; and $f(\perp,y) = f_1(y)$. Then the winning condition can be expressed as
\begin{align}
    \begin{cases}
        a \oplus b = f(x,y) & \mbox{if $x \neq \perp$}\\
        \phantom{a \oplus \hspace*{2mm}} b = f(x,y) & \mbox{if $x = \perp$.}
    \end{cases}
\end{align}

As mentioned in~\cref{sec:intro} (see \cref{eq:tilted-chsh}), a game of this form was introduced by Acín, Massar, and Pironio~\cite{AMP12} where, with some probability, the verifier plays the CHSH game and, with some probability, the verifier sends question 0 to one player and requires answer 0 to win.
Subsequently, this became known as the \emph{tilted CHSH game} (for example,~\cite{BP15b} use this terminology).

This motivates the following definition.

\begin{definition}[tilted XOR game]
	A \emph{tilted XOR game} is a nonlocal game $G=(X,Y,A,B,\pi,V)$ where $A=B=\{0,1\}$ and $X$ contains a distinguished element $\perp$, and there exists a function $f:X\times Y\rightarrow\{0,1\}$ such that
	\begin{align*}
		V(a,b|x,y)=
        \begin{cases}
        \delta_{a \oplus b,\,f(x,y)} & \text{if $x \neq \perp$}\\
        \ \delta_{b,\,f(x,y)} & \text{if $x = \perp$.}
        \end{cases}
	\end{align*}
\end{definition}

By swapping the roles of Alice and Bob, we can see that the variant of tilted XOR games where Bob receives the distinguished question (in  which case Alice's answer bit is relevant) is completely equivalent.
And one may generalise tilted XOR games to a \emph{two-sided} variant where both Alice and Bob may receive the distinguished question.
The I3322 game%
\footnote{These games are usually presented via their Bell inequality presentations, amounting to rescalings of the bias.}
\cite{Fro81,Sli03,CG04} provides an example of the two-sided variant. 
We do not make use of the two-sided variant in this work.

An alternative way of thinking about tilted XOR games is as XOR games---where $a \oplus b = f(x,y)$ must hold in all cases---with a condition added that $a = 0$ must also hold whenever $x = \perp$.
In other words, for the $x = \perp$ case, 
\begin{align}
    V(a,b|\perp,y)=\delta_{a\oplus b,f(\perp,y)}\delta_{a,0}.
\end{align}
This perspective enables us to make use of certain results pertaining to XOR games in our analysis.

We can express the value of a quantum strategy $p(a,b|x,y)=\braket{\psi}{P^x_a\otimes Q^y_b}{\psi}$ for $G$ using the associated quantum correlation $C(x,y)=\braket{\psi}{A_x\otimes B_y}{\psi}$ with $A_x=P^x_0-P^x_1$ and $B_y=Q^y_0-Q^y_1$. Without loss of generality, we may suppose that $A_\perp=I$, since Alice's answer on the distinguished question does not affect the value. Then, the value of $p$ is
\begin{align*}
	\omega(G,p)&=\sum_{y\in Y}\pi(\perp,y)\braket{\psi}{I\otimes Q^y_{f(\perp,y)}}{\psi}+\sum_{\substack{x\in X\backslash\{\perp\},y\in Y\\a,b\in\{0,1\}:\,a+b=f(x,y)}}\pi(x,y)\braket{\psi}{P^x_{a}\otimes Q^y_b}{\psi}\\
    &=\frac{1}{2}\sum_{y\in Y}\pi(\perp,y)\braket{\psi}{A_\perp\otimes(I+(-1)^{f(\perp,y)}B_y)}{\psi}\\
    &\qquad+\frac{1}{4}\sum_{x\in X\backslash\{\perp\},y\in Y}\pi(x,y)\sum_{a\in\{0,1\}}\braket{\psi}{(I+(-1)^aA_x)\otimes(I+(-1)^{a+f(x,y)}B_y)}{\psi}\\
    &=\frac{1}{2}+\frac{1}{2}\sum_{x\in X,y\in Y}(-1)^{f(x,y)}\pi(x,y)\braket{\psi}{A_x\otimes B_y}{\psi}.
\end{align*}
As for an XOR game, let the game matrix be $H_{x,y}=(-1)^{f(x,y)}\pi(x,y)$ and the bias of a correlation $\beta(G,C)=\sum_{x\in X,y\in Y}H_{x,y}C(x,y)$.

\begin{definition}
    The \emph{XOR relaxation} of a tilted XOR game $G$ is the XOR game, that we denote as  $G_{\xor}$, with the same distribution where $V_{\xor}(a,b|x,y)=\delta_{a\oplus b,f(x,y)}$ for all $x \in X$ and $y \in Y$. That is, the distinguished question $\perp$ is converted to an XOR type question.
\end{definition}

\begin{lemma}\label{lem:tilted-xor-bound}
    Let $G$ be a tilted XOR game. Then, $\omega(G_{\xor})=\omega(G)$ and $\omega^\ast(G_{\xor})\geq\omega^\ast(G)$.
\end{lemma}

\begin{proof}
    Since $G$ is the game $G_{\xor}$ with additional winning conditions, we can immediately deduce $\omega(G_{\xor}) \ge \omega(G)$
    and $\omega^\ast(G_{\xor}) \ge \omega^\ast(G)$.
    
    For $\omega(G_{\xor}) \le \omega(G)$, consider any optimal deterministic strategy $p(a,b|x,y)=\delta_{g(x),a}\delta_{h(y),b}$ for $G_{\xor}$.
    Then, define the classical strategy $p'$ for $G$ via the functions $g'(x)=g(x) \oplus g(\perp)$ and $h'(y)=h(y) \oplus g(\perp)$. Since $g'(x) \oplus h'(y)=g(x) \oplus h(y)$ and $g'(\perp)=0$, this is a strategy for the tilted XOR game $G$ with the same value as the value of $p$ for $G_{\xor}$. So, the optimal classical value of $G$ upper bounds the optimal classical value of $G_{\xor}$.
\end{proof}

Now we define decision problems associated with approximating the value of tilted XOR games.
Since any tilted XOR game can be trivially won with probability at least $\frac{1}{2}$ by Bob simply outputting a random bit, the range of interesting success probabilities is $(\frac{1}{2},1]$.

\begin{definition}
    For $\frac{1}{2} \le s \le c \le 1$, define $\txor_{c,s}$ as the problem of deciding, for a given tilted XOR game $G$, if $\omega(G)\geq c$ or $\omega(G)<s$, with the promise that one of these two holds.
    Similarly, $\txor_{c,s}^\ast$ is the problem of deciding, for a given tilted XOR game $G$, if $\omega^\ast(G)\geq c$ or $\omega^\ast(G)<s$, with the promise that one of these two holds.
    Also, the succinctly-presented versions of these problems are denoted as $\suctxor_{c,s}$ and $\suctxor_{c,s}^\ast$.
\end{definition}

By~\cref{lem:tilted-xor-bound}, the computational problems $\txor_{c,s}$ and $\xor_{c,s}$ are equivalent.
Therefore, the result in~\cite{Has01} carries over to $\txor_{c,s}$.
Namely that, for $c=\frac{3}{4}-\varepsilon$ and $s=\frac{11}{16}+\varepsilon$, $\txor_{c,s}$ is $\tsf{NP}$-complete (for any $\varepsilon > 0$). Also, for $c>s$, $\txor_{c,s}^\ast\in\tsf{RE}$ as the problem consists of approximating the quantum value of a nonlocal game to constant precision.

\subsection{The tilted cube game}\label{sec:tilted-cube}

In this section, we study the structure of a tilted XOR game arising from the cube game, which we will make use of as a gadget to reduce from E3-LIN.

\begin{definition}[tilted cube game $G^{\tiltedsquare}$]
    The \emph{tilted cube game}, that we denote as $G^{\tiltedsquare}$, is the XOR game $G^{\Box}$ from \cref{def:cube} with a tilt added where the distinguished variable is $\perp = 000$.
\end{definition}

Thus, for questions of the form $(\perp,y)$, the winning condition is that Bob's answer bit is: 1, if $y \in \{001,010,100\}$; and 0 if $y = 111$.
Equivalently, it is the cube game, but where Alice's answer bit to the special question $\perp$ is always deemed $0$.
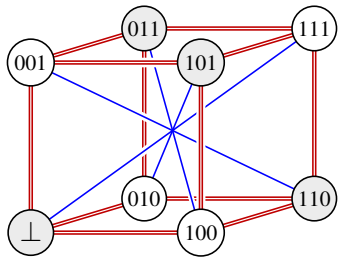
\begin{figure}[h!]
\centering
\scalebox{0.75}{
\begin{tikzpicture}[
shape_style/.style={draw=black!90!black, fill=white!10, thick, minimum size=3mm},
shape_style-g/.style={draw=black!90!black, fill=gray!15, thick, minimum size=3mm},
label_style/.style={font=\small\sffamily, text=black}
]
\draw[double,thick,darkred] (0,0) -- (3,0);
\draw[double,thick,darkred] (0,0) -- (0,3);
\draw[double,thick,darkred] (0,3) -- (3,3);
\draw[double,thick,darkred] (2.0,0.6) -- (5.0,0.6);
\draw[double,thick,darkred] (5.0,0.6) -- (5.0,3.6);
\draw[double,thick,darkred] (2.0,0.6) -- (2.0,3.6);
\draw[double,thick,darkred] (2.0,3.6) -- (5.0,3.6);
\draw[double,thick,darkred] (0,0) -- (2.0,0.6);
\draw[double,thick,darkred] (3,0) -- (5.0,0.6);
\draw[double,thick,darkred] (0,3) -- (2.0,3.6);
\draw[double,thick,darkred] (3,3) -- (5.0,3.6);
\draw[thick,blue] (2.0,0.6) -- (3.0,3.0);
\draw[line width=1.0mm,white] (3,0) -- (2.55,1.62);
\draw[thick,blue] (3,0) -- (2.0,3.6);
\draw[line width=1.3mm,white] (0,3) -- (2.25,1.92);
\draw[thick,blue] (0,3) -- (5.0,0.6);
\draw[line width=1.3mm,white] (0,0) -- (2.25,1.62);
\draw[thick,blue] (0,0) -- (5.0,3.6);
\draw[line width=1.7mm,white] (0,3) -- (3,3);
\draw[double,thick,darkred] (0,3) -- (3,3);
\draw[line width=1.6mm,white] (3,0) -- (3,3);
\draw[double,thick,darkred] (3,0) -- (3,3);
\node[shape_style-g, circle, minimum size=3mm, inner sep=0.3pt] (circ) at (0,0) {\small $\phantom{A_{1...}}$};\node at (0,0.03) {\large $\perp$};
\node[shape_style, circle, minimum size=3mm, inner sep=2pt] (circ) at (3,0) {\small 100};
\node[shape_style, circle, minimum size=3mm, inner sep=2pt] (circ) at (0,3) {\small 001};
\node[shape_style-g, circle, minimum size=3mm, inner sep=2pt] (circ) at (3,3) {\small 101};
\node[shape_style, circle, minimum size=3mm, inner sep=1.6pt] (circ) at (2.0,0.6) {\small 010};
\node[shape_style-g, circle, minimum size=3mm, inner sep=1.6pt] (circ) at (5.0,0.6) {\small 110};
\node[shape_style-g, circle, minimum size=3mm, inner sep=1.6pt] (circ) at (2.0,3.6) {\small 011};
\node[shape_style, circle, minimum size=3mm, inner sep=1.6pt] (circ) at (5.0,3.6) {\small 111};
\end{tikzpicture}
}
\caption{The tilted cube game as a bipartite graph. For question~$\perp$, Alice's answer bit is deemed 0.}
\label{fig:tilted-cube-game}
\end{figure}

It is clear from this definition that $(G^{\tiltedsquare})_{\text{XOR}}=G^{\square}$.
Therefore, it follows from \cref{lem:tilted-xor-bound} that $\omega(G^{\tiltedsquare})=\omega(G^{\Box})=\frac{3}{4}$.
It turns out that $\omega^\ast(G^{\tiltedsquare}) = \omega^\ast(G^{\Box}) = \frac{3}{4}$ also holds.

\begin{lemma}
    $\omega^\ast(G^{\tiltedsquare}) = \omega^\ast(G^{\Box}) = \frac{3}{4}$. 
\end{lemma}

\begin{proof}
    It is clear that $\omega^\ast(G^{\Box})\geq\omega(G^{\Box})=\frac{3}{4}$. We use the dual SDP (Eq.~\eqref{eq:dual} discussed in \cref{sec:prelim-xor}) to upper bound the quantum value of $G^{\Box}$. Ordering the elements of $X_0$ as $(000,011,101,110)$ and the elements of $X_1$ as $(111,100,010,001)$ the game matrix is
    \begin{align*}
        H=\frac{1}{16}\begin{bmatrix}+1&-1&-1&-1\\-1&+1&-1&-1\\-1&-1&+1&-1\\-1&-1&-1&+1\end{bmatrix}.
    \end{align*}
    Taking the square, we find that
    \begin{align*}
        H^2=\frac{1}{16^2}\begin{bmatrix}4&0&0&0\\0&4&0&0\\0&0&4&0\\0&0&0&4\end{bmatrix}=\frac{1}{64}I,
    \end{align*}
    which implies that the eigenvalues of $H$ are $\pm\frac{1}{8}$. Then, the eigenvalues of $B=\frac{1}{2}\squ*{\begin{smallmatrix}0&H\\H^T&0\end{smallmatrix}}$ are $\pm\frac{1}{16}$. Therefore, $\frac{1}{16}I\geq B$, giving that $a_x=b_y=\frac{1}{8}$ is a feasible point for the dual SDP (Eq.~\eqref{eq:dual}) for $G^{\Box}$, giving an upper bound on the bias $\beta^\ast(G^{\Box})\leq\frac{1}{2}\sum_{x\in X_0}a_x+\frac{1}{2}\sum_{y\in X_1}b_y=\frac{1}{2}$. This gives the upper bound on the value $\omega^\ast(G^{\Box})\leq\frac{1}{2}+\frac{1}{2}\cdot\frac{1}{2}=\frac{3}{4}$.

    Since $\omega^{\ast}(G^{\tiltedsquare})$ is sandwiched between $\omega^\ast(G^{\Box})$ and $\omega(G^{\tiltedsquare})$, this implies $\omega^{\ast}(G^{\tiltedsquare})=\frac{3}{4}$.
\end{proof}

It can be seen from the proof above that the marginal row and column biases are all $\alpha=\frac{1}{8}$.
Hence, \cref{thm:xor-rigidity-exact} applies to $G^{\Box}$.
Next, in~\cref{lem:optimal-case}, we show that winning $G^{\tiltedsquare}$ optimally implies that Alice's observables satisfy a linear constraint.
Later, in~\cref{lem:near-optimal}, we will show an approximate version of this, for near-optimal strategies.

\begin{lemma}\label{lem:optimal-case}
    Let $C(x,y)=\braket{\psi}{A_x\otimes B_y}{\psi}$ be an optimal correlation for $G^{\tiltedsquare}$. Suppose that the observables are unitaries. Then, $A_{011}$, $A_{101}$, and $A_{110}$ preserve the support of $\ket{\psi}$ on Alice's system and, restricted to this space, commute and satisfy $A_{011}A_{101}A_{110}=I$.
\end{lemma}

\begin{proof}
    Since $\omega^\ast(G^{\Box})=\omega^{\ast}(G^{\tiltedsquare})$, the correlation $C$ is also an optimal correlation for $G^{\Box}$. Let $\ket{u_x}=(A_x\otimes I)\ket{\psi}$ and $\ket{v_y}=(I\otimes B_y)\ket{\psi}$. Since the optimal value and the Frobenius norm are equal in terms of the usual and the real inner product, we have by \cref{thm:xor-rigidity-exact} that $\sum_{x\in X}H_{x,y}\ket{u_x}=\frac{1}{8}\ket{v_y}$ and $\sum_{y\in Y}H_{x,y}\ket{v_y}=\frac{1}{8}\ket{u_x}$, and therefore
    \begin{align*}
         &\sum_{x}H_{x,y}(A_x\otimes I)\ket{\psi}=\frac{1}{8}(I\otimes B_y)\ket{\psi},\\
         &\sum_{y}H_{x,y}(I\otimes B_y)\ket{\psi}=\frac{1}{8}(A_x\otimes I)\ket{\psi}.
    \end{align*}
    Taking the square,
    \begin{align*}
        \frac{1}{64}\ket{\psi}=\frac{1}{64}(I\otimes B_y)^2\ket{\psi}=\frac{1}{8}\sum_{x}H_{x,y}(A_x\otimes B_y)\ket{\psi}=\parens[\Big]{\sum_{x}H_{x,y}(A_x\otimes I)}^2\ket{\psi}.
    \end{align*}
    Now, note that for any $\ket{v}$ in the support of $\ket{\psi}$ on $H_A$, there exists $\ket{w}\in H_B$ such that $\ket{v}=(I\otimes\bra{w})\ket{\psi}$. Therefore, $A_x\ket{v}=8\sum_{y}H_{x,y}(I\otimes \bra{w}B_y)\ket{\psi}$, which remains in the support of $\ket{\psi}$ on $H_A$. Thus, the $A_x$ preserve this space. In the following, we restrict to the support of $\ket{\psi}$ on $H_A$.
    
    With the above restriction, $\parens*{\sum_xH_{x,y}A_x}^2=\frac{1}{64}I$.
    It follows that, for all $y \in X_1$,
    \begin{align*}
        (A_{y \oplus 111}-A_{y \oplus 100}-A_{y \oplus 010}-A_{y\oplus 001})^2=4I,
    \end{align*}
    as $|H_{x,y}|=\frac{1}{16}$. Since $C$ is a strategy for the tilted XOR game, we know that $A_{000}=I$. Writing $A_1=A_{011}$, $A_2=A_{101}$, and $A_3=A_{110}$, we get the four relations
    \begin{align*}
        (+I-A_1-A_2-A_3)^2&=4I,\\
        (-I+A_1-A_2-A_3)^2&=4I,\\
        (-I-A_1+A_2-A_3)^2&=4I,\\
        (-I-A_1-A_2+A_3)^2&=4I.
    \end{align*}
    Adding the first two relations gives
    \begin{align*}
        8I&=((I-A_1)-(A_2+A_3))^2+(-(I-A_1)-(A_2+A_3))^2\\
        &=2\parens*{(I-A_1)^2+(A_2+A_3)^2}\\
        &=2\parens*{I-2A_1+I+I+A_2A_3+A_3A_2+I}\\
        &=2\parens*{4I-2A_1+\{A_2,A_3\}},
    \end{align*}
    which can be rearranged to give $\{A_2,A_3\}=2A_1$. By symmetry of the relations, we get the other anticommutation relations $\{A_1,A_2\}=2A_3$ and $\{A_1,A_3\}=2A_2$. Consider the following:
    \begin{align*}
        4A_1&=\{A_2,2A_3\}\\
        &=\set*{A_2,\{A_1,A_2\}}\\
        &=A_2(A_1A_2+A_2A_1)+(A_1A_2+A_2A_1)A_2\\
        &=2\parens*{A_2A_1A_2+A_1}.
    \end{align*}
    Rearranging, we find $A_2A_1A_2-A_1=0$, and then multiplying by $A_2$ gives the commutation relation $[A_1,A_2]=0$. In the same way, we find the other commutation relations $[A_1,A_3]=[A_2,A_3]=0$. Hence, Alice's observables commute. To finish the proof, we see using the commutation relations that
    \begin{align*}
        &A_1A_2A_3=\frac{1}{2}A_1\{A_2,A_3\}=A_1^2=I.\qedhere
    \end{align*}
\end{proof}

We will also require a \emph{robust} version of~\cref{lem:optimal-case} for approximately optimal strategies for~$G^{\tiltedsquare}$.
The robust version is \cref{lem:near-optimal}, deferred to~\cref{sec:proof-gapped}.

\section{Hardness of tilted XOR games}\label{sec:hardness-tilted-xor}\label{sec:hardness-of-tilted-XOR}

In this section, we state and prove the main results of this work.
The first result concerns the problem of determining whether the quantum value of a tilted XOR game is $\ge \frac{3}{4}$ or $< \frac{3}{4}$.

\begin{theorem}[ungapped version]\label{thm:main-exact}
    The problem $\txor_{\scriptscriptstyle \frac{3}{4},\frac{3}{4}}^\ast$ is undecidable.
\end{theorem}
Note that this is in contrast to the problem $\xor_{\scriptscriptstyle \frac{3}{4},\frac{3}{4}}^\ast$, which is decidable in $\tsf{PSPACE}$, due to~\cite{Can88}.
Beyond being of independent interest, \cref{thm:main-exact} and its proof (in \cref{sec:proof-ungapped}) provide a warm-up for the more technically involved proof of the gapped version.

\begin{theorem}[gapped version]\label{thm:main-robust}
    There exist constants $c > s$ such that $\txor_{c,s}^\ast$ is $\tsf{RE}$-complete.
\end{theorem}

\cref{thm:main-robust} follows immediately from its succinct version.

\begin{theorem}[succinct gapped version]\label{thm:main-robust-succinct}
    There exist constants $c > s$ such that $\txormip_{c,s}^\ast$ is $\tsf{RE}$-complete in the following sense: there is a polynomial-time reduction from any decision problem in $\tsf{RE}$ to two-prover interactive proof systems which play a tilted XOR game, with completeness and soundness probabilities $c$ and $s$.
\end{theorem}

From the proof of ~\cref{thm:main-robust-succinct} (in~\cref{sec:proof-gapped}), we may choose explicit values $s=\frac{3}{4}-10^{-8}$ and $c=\frac{3}{4}-\varepsilon$, for any $\varepsilon \in (0,10^{-8})$, in both~\cref{thm:main-robust,thm:main-robust-succinct}.

\subsection{Proof of hardness of ungapped version of tilted XOR games}\label{sec:proof-ungapped}

\begin{proof}[Proof of \cref{thm:main-exact}]
It is known that $\ethreelin_{1,1}^*$ is undecidable~\cite{Slo19}.
Our reduction from $\ethreelin_{1,1}^*$ to $\txor_{\scriptscriptstyle \frac{3}{4},\frac{3}{4}}^\ast$is based on H\aa stad's reduction (reviewed in \cref{sec:cube}), but relies on a different interpretation of the resulting game. 
Given an E3-LIN game $G$, let $\widetilde{G}$ denote the game produced by the gadgets in~\cref{fig:tilted-cube-game-reduction}.
Rather than viewing $\widetilde{G}$ as an XOR game, we reinterpret it as a tilted XOR game with distinguished input $\perp$.
We will show that $\omega^*(\widetilde{G}) = \frac{3}{4}$ if and only if the E3-LIN game has a perfect quantum strategy.

First, suppose that a strategy for the tilted XOR game $\widetilde{G}$ attains success probability $\frac{3}{4}$.
Then, since $\widetilde{G}$ is a probabilistic mixture of edge-disjoint tilted cube games, each of which has maximum success probability $\frac{3}{4}$, it follows that the strategy achieves success probability $\frac{3}{4}$ for each of the tilted cube games.
Therefore, by \cref{lem:optimal-case}, for each equation of the form $v_i \oplus v_j \oplus v_k = b$, the corresponding observables $A_i$, $A_j$, and $A_k$ commute and satisfy $A_iA_jA_k = (-1)^bI$ (when restricted to the support of the shared state on Alice's system).
It follows that the observables corresponding to the variables in the E3-LIN instance are an operator solution to the system of equations, which can be turned into a perfect strategy for the E3-LIN game.

Next, suppose that the E3-LIN game has a perfect quantum strategy. Then, by~\cite{CM14a}, there exists a $d$-dimensional operator solution to the E3-LIN game, which consists of a binary observable $A_i$ assigned to each variable $v_i$, where the observables satisfy the following property: 
for each equation $v_i \oplus v_j \oplus v_k = b$ in the E3-LIN game, the observables $A_i$, $A_j$, and $A_k$ commute and $A_iA_jA_k = (-1)^bI$.
This can be turned into a strategy for the tilted XOR game that succeeds with probability $\frac{3}{4}$ as follows.
For each cube within the tilted XOR game $\widetilde{G}$, assign the operators $A_i$, $A_j$, and $A_k$ to the vertices $i$, $j$, and $k$ (respectively) and assign the operator $I$ to vertex $\perp$.
For the remaining (odd parity) vertices of the cube, assign the binary observables
\medskip
\begin{align*}
    B_{111} &= \textstyle{\frac{1}{2}}(+(-1)^bI - A_i^T - A_j^T - A_k^T) \\[1mm]
    B_{100} &= \textstyle{\frac{1}{2}}(-(-1)^bI + A_i^T - A_j^T - A_k^T) \\[1mm]
    B_{010} &= \textstyle{\frac{1}{2}}(-(-1)^bI - A_i^T + A_j^T - A_k^T) \\[1mm]
    B_{001} &= \textstyle{\frac{1}{2}}(-(-1)^bI - A_i^T - A_j^T + A_k^T).
\end{align*}
Then it is a straightforward exercise to show the above are proper binary observables (\textit{i.e.} hermitian unitaries) and, with respect to the entangled state
\begin{align}
  \ket{\psi} = \frac{1}{\sqrt{d}} \sum_{k=1}^d\ket{k}\otimes\ket{k},
\end{align}
the success probability of this strategy for the tilted XOR game $\widetilde{G}$ succeeds with probability $\frac{3}{4}$, where we use the fact that
\begin{align}
    C(x,y)
    =
    \bra{\psi}A_x\otimes B_y\ket{\psi}
    =
    \frac{1}{d}\Tr\bigl(A_x^T B_y\bigr).
\end{align}
\end{proof}

Recall that \cref{thm:main-exact} applies to quantum strategies, which use a finite-dimensional tensor product of entanglement.
In \cref{sec:commuting-operator} we consider strategies in a different model, called the commuting-operator strategies, which use a model of entanglement that differs in infinite dimensions.
In \cref{thm:commuting-operator}, we show a similar result to \cref{thm:main-exact} for commuting-operator strategies.

\subsection{Proof of hardness of gapped version of tilted XOR games}\label{sec:proof-gapped}

\subsubsection{Improved XOR game structure theorem}

In this section, we prove a tighter average-case version of \cref{thm:xor-rigidity}~\cite[Theorem 3.1]{Slo11} when there is a symmetry condition on the XOR game.

\begin{lemma}\label{lem:symmetric-xor}
	Let $G$ be an XOR game such that $|X|=|Y|$, and suppose that all the optimal row and column biases are equal $a_x=b_y\eqqcolon\alpha$. Then, for any $\varepsilon$-optimal vector correlation $C(x,y)=\braket{u_x}{v_y}$,
	\begin{align*}
		&\sum_{x\in X}\norm[\Big]{\sum_{y\in Y}H_{x,y}\ket{v_y}-\alpha\ket{u_x}}^2\leq 4(\alpha+4\beta^\ast(G))\varepsilon
	\end{align*}
    By symmetry, we also have that
    \begin{align*}
		&\sum_{y\in Y}\norm[\Big]{\sum_{x\in X}H_{x,y}\ket{u_x}-\alpha\ket{v_y}}^2\leq 4(\alpha+4\beta^\ast(G))\varepsilon
	\end{align*}
\end{lemma}

\begin{proof}
	First, write $n=|X|$ and  $\beta=\beta^\ast(G)$; we know $\beta=n\alpha$. By definition, there exists some $\varepsilon'\leq2\varepsilon$ such that $\sum_{x,y}H_{x,y}\braket{u_x}{v_y}=\beta-\varepsilon'$. Then, we expand the left-hand side from the lemma statement to get that
	\begin{align*}
		\sum_x\norm[\Big]{\sum_yH_{x,y}\ket{v_y}-\alpha\ket{u_x}}^2&=\sum_x\parens[\Big]{\alpha^2-2\alpha\sum_yH_{x,y}\braket{u_x}{v_y}+\norm[\Big]{\sum_yH_{x,y}\ket{v_y}}^2}\\
		&=-\alpha\beta+2\alpha\varepsilon'+\sum_x\norm[\Big]{\sum_yH_{x,y}\ket{v_y}}^2
	\end{align*}
	Now, we focus on bounding the last term. Following the argument of \cite[Theorem 3.1]{Slo11}, let $\ket{u_x'}$ be the normalisation of $\sum_yH_{x,y}\ket{v_y}$ (if $\sum_yH_{x,y}\ket{v_y}=0$, we take $\ket{u'_x}$ to be an arbitrary unit vector). Since
    \begin{align*}
        \sum_{x,y}H_{x,y}\braket{u_x'}{v_y}=\sum_{x}\norm[\Big]{\sum_yH_{x,y}\ket{v_y}}\geq\sum_{x,y}H_{x,y}\braket{u_x}{v_y},
    \end{align*}
    the vector correlation $C'(x,y)=\braket{u_x'}{v_y}$ is an $\varepsilon$-optimal strategy as well, with bias $\beta-\varepsilon'_B$ for some $\varepsilon'_B\leq\varepsilon'\leq 2\varepsilon$. Then,
	\begin{align*}
		\sum_x\norm[\Big]{\sum_yH_{x,y}\ket{v_y}}^2&=\sum_x\norm[\Big]{\sum_yH_{x,y}\ket{v_y}-\alpha\ket{u_x'}+\alpha\ket{u_x'}}^2\\
		&=\sum_x\parens[\Big]{\alpha^2+2\alpha\sum_yH_{x,y}\braket{u_x'}{v_y}-2\alpha^2+\norm[\Big]{\sum_yH_{x,y}\ket{v_y}-\alpha\ket{u_x'}}^2}\\
		&=\alpha\beta-2\alpha\varepsilon'_B+\sum_x\parens[\Big]{\sum_yH_{x,y}\braket{u_x'}{v_y}-\alpha}^2.
	\end{align*}
	As in the SDP formulation, write $M$ for the $2n\times 2n$ Gram matrix of the vectors $\ket{u_x'}$ and $\ket{v_y}$ over $x$ and $y$, and let
	$$S=\frac{1}{2}\begin{bmatrix}\alpha I&-H\\-H^T&\alpha I\end{bmatrix}.$$
	We have that $\varepsilon'_B=\Tr(SM)$, $\alpha I\geq S\geq 0$, and $\norm{M}\leq\Tr(M)\leq 2n$. Then, 
	\begin{align*}
		\sum_x\parens[\Big]{\sum_yH_{x,y}\braket{u_x'}{v_y}-\alpha}^2&=4\sum_{x\in X}(SM)_{x,x}^2\\
		&\leq4\norm{SM}_F^2=4\Tr(SM^2S)\\
		&\leq4\norm{S}\norm{M}\Tr(SM)\\
		&\leq4(\alpha)(2n)\varepsilon'_B=8\beta\varepsilon'_B.
	\end{align*}
	Putting it all together,
	\begin{align*}
		&\sum_x\norm[\Big]{\sum_yH_{x,y}\ket{v_y}-\alpha\ket{u_x}}^2\leq-\alpha\beta+2\alpha\varepsilon'+(\alpha\beta-2\alpha\varepsilon'_B+8\beta\varepsilon'_B)\\
		&\leq4(\alpha+4\beta)\varepsilon.\qedhere
	\end{align*}
\end{proof}

Now, we use the above to prove a robust version of the \cref{lem:optimal-case}.

\begin{proposition}\label{lem:near-optimal}
    Let $C(x,y)=\braket{\psi}{A_x\otimes B_y}{\psi}$ be an $\varepsilon$-optimal quantum correlation for $G^{\tiltedsquare}$ such that the observables are unitaries. Then,
    \begin{align*}
        \norm[\big]{\parens[\big]{(A_{011}A_{101}A_{110}-I)\otimes I}\ket{\psi}}\leq 68\sqrt{13\varepsilon}.
    \end{align*}
\end{proposition}

By symmetry, the same upper bound holds independent of the order of the observables.

\begin{proof}
    Note that Alice's observable $A_{000}=I$. Then, as $\alpha=\frac{1}{8}$ and $\beta^\ast(G^{\tiltedsquare})=\frac{3}{4}$, \cref{lem:symmetric-xor} implies that
    \begin{align}\label{eq:basic-rigidity}
        \sum_y\norm[\Big]{\sum_{x}H_{x,y}(A_x\otimes I)\ket{\psi}-\frac{1}{8}(I\otimes B_y)\ket{\psi}}^2\leq4\parens*{\frac{1}{8}+4\frac{3}{4}}\varepsilon\leq13\varepsilon.
    \end{align}
    Now, let $\ket{\psi}=\sum_{i}\sqrt{p_i}\ket{a_i}\otimes\ket{b_i}$ be the Schmidt decomposition of $\ket{\psi}$, and take $\lambda:H_B\rightarrow H_A$ to be the operator $\lambda=\sum_i\sqrt{p_i}\ketbra{a_i}{b_i}$. Then, $\lambda\lambda^\ast=\psi_A$ and $\lambda^\ast\lambda=\psi_B$, the marginals of $\ket{\psi}$ on $H_A$ and $H_B$, respectively. We can rewrite \eqref{eq:basic-rigidity} in terms of the Frobenius norm as
    \begin{align*}
        \sum_y\norm[\Big]{\sum_{x}H_{x,y}A_x\lambda-\frac{1}{8}\lambda\overline{B}_y}^2_F\leq13\varepsilon,
    \end{align*}
    where the complex conjugate is with respect to the basis $\{\ket{b_i}\}_i$. Writing $h_{x,y}=16H_{x,y}$, we have $|h_{x,y}|=1$ and
    \begin{align*}
        \sum_y\norm[\Big]{\sum_{x}h_{x,y}A_x\lambda-2\lambda\overline{B}_y}^2_F\leq2^{8}\cdot13\varepsilon.
    \end{align*}
    Next, we remove the dependence on $\overline{B}_y$:
    \begin{align*}
        \norm[\Big]{\parens[\Big]{\sum_{x}h_{x,y}A_x}^2\lambda-4\lambda}_F&\leq\norm[\Big]{\parens[\Big]{\sum_{x}h_{x,y}A_x}^2\lambda-2\sum_{x}h_{x,y}A_x\lambda\overline{B}_y}_F+\norm[\Big]{2\sum_{x}h_{x,y}A_x\lambda\overline{B}_y-4\lambda}_F\\
        &\hspace{-0.5cm}\leq\norm[\Big]{\sum_{x}h_{x,y}A_x}\norm[\Big]{\sum_{x}h_{x,y}A_x\lambda-2\lambda\overline{B}_y}_F+2\norm[\Big]{\sum_{x}h_{x,y}A_x\lambda-2\lambda\overline{B}_y}_F\norm[\Big]{\overline{B}_y}\\
        &\hspace{-0.5cm}\leq6\norm[\Big]{\sum_{x}h_{x,y}A_x\lambda-2\lambda\overline{B}_y}_F,
    \end{align*}
    and therefore
    \begin{align*}
        \sum_y\norm[\Big]{\parens[\Big]{\sum_{x}h_{x,y}A_x}^2\lambda-4\lambda}_F^2\leq36\sum_y\norm[\Big]{\sum_{x}h_{x,y}A_x\lambda-2\lambda\overline{B}_y}_F^2\leq2^{10}\cdot117\varepsilon.
    \end{align*}
    Now, as in the proof of \cref{lem:optimal-case}, let $A_1=A_{011}$, $A_2=A_{101}$, and $A_3=A_{110}$. Then we expand to see that
    \begin{align*}
        &\norm*{(I-A_1-A_2-A_3)^2\lambda-4\lambda}_F^2+\norm*{(-I+A_1-A_2-A_3)^2\lambda-4\lambda}_F^2\\
        &\qquad+\norm*{(-I-A_1+A_2-A_3)^2\lambda-4\lambda}_F^2+\norm*{(-I-A_1-A_2+A_3)^2\lambda-4\lambda}_F^2\leq 2^{10}\cdot117\varepsilon.
    \end{align*}
    Using the facts that
    \begin{align*}
        (I-A_1-A_2-A_3)^2+(-I+A_1-A_2-A_3)^2&=2\parens*{4I-2A_1+\{A_2,A_3\}}\\
        (-I-A_1+A_2-A_3)^2+(-I-A_1-A_2+A_3)^2&=2\parens*{4I+2A_1-\{A_2,A_3\}},
    \end{align*}
    we find that
    \begin{align*}
        \norm*{\{A_2,A_3\}\lambda-2A_1\lambda}_F^2&=\frac{1}{8}\norm*{(I-A_1-A_2-A_3)^2\lambda+(-I+A_1-A_2-A_3)^2\lambda-8\lambda}_F^2\\
        &\qquad+\frac{1}{8}\norm*{(-I-A_1+A_2-A_3)^2\lambda+(-I-A_1-A_2+A_3)^2\lambda-8\lambda}_F^2\\
        &\hspace{-1.6cm}\leq\frac{1}{4}\norm*{(I-A_1-A_2-A_3)^2\lambda-4\lambda}_F^2+\frac{1}{4}\norm*{(-I+A_1-A_2-A_3)^2\lambda-4\lambda}_F^2\\
        &\hspace{-1.6cm}\qquad+\frac{1}{4}\norm*{(-I-A_1+A_2-A_3)^2\lambda-4\lambda}_F^2+\frac{1}{4}\norm*{(-I-A_1-A_2+A_3)^2\lambda-4\lambda}_F^2\\
        &\hspace{-1.6cm}\leq 2^{8}\cdot117\varepsilon.
    \end{align*}
    By symmetry, the other anticommutation relations hold approximately as well, with the same error. Next, let $\widetilde{B}_1=\frac{1}{2}\parens*{\overline{B}_{100}-\overline{B}_{010}-\overline{B}_{001}-\overline{B}_{111}}$, $\widetilde{B}_2=\frac{1}{2}\parens*{\overline{B}_{010}-\overline{B}_{100}-\overline{B}_{111}-\overline{B}_{001}}$, and $\widetilde{B}_3=\frac{1}{2}\parens*{\overline{B}_{001}-\overline{B}_{111}-\overline{B}_{100}-\overline{B}_{010}}$. By construction, $\norm{\widetilde{B}_i}\leq 2$ and by \cref{lem:symmetric-xor}, $\sum_{i}\norm{A_i\lambda-\lambda\widetilde{B}_i}_F^2\leq 2^{6}\cdot13\varepsilon$. Now, note that
    \begin{align*}
        A_1A_2A_3-I&=\frac{1}{2}\parens*{A_1(\{A_2,A_3\}-2A_1)+A_1[A_2,A_3]}\\
        &\hspace{-0.5cm}=\frac{1}{2}A_1(\{A_2,A_3\}-2A_1)-\frac{1}{4}A_1A_2\parens*{\{A_2,\{A_2,A_3\}\}-4A_3}\\
        &\hspace{-0.5cm}=\frac{1}{4}A_1(\{A_2,A_3\}-2A_1)-\frac{1}{2}A_1A_2\parens*{\{A_2,A_1\}-2A_3}-\frac{1}{4}A_1A_2(\{A_2,A_3\}-2A_1)A_2
    \end{align*}
   
    Then, we can use this to bound
    \begin{align*}
        \norm*{(A_1A_2A_3-I)\lambda}_F&\leq\frac{1}{4}\norm*{A_1(\{A_2,A_3\}-2A_1)\lambda}_F+\frac{1}{2}\norm*{A_1A_2\parens*{\{A_2,A_1\}-2A_3}\lambda}_F\\
        &+\frac{1}{4}\norm*{A_1A_2(\{A_2,A_3\}-2A_1)A_2\lambda}_F\\
        &\leq\frac{1}{4}\norm*{(\{A_2,A_3\}-2A_1)\lambda}_F+\frac{1}{2}\norm*{\parens*{\{A_2,A_1\}-2A_3}\lambda}_F\\
        &+\frac{1}{4}\norm*{(\{A_2,A_3\}-2A_1)\lambda\widetilde{B}_2}_F+\frac{1}{4}\norm*{(\{A_2,A_3\}-2A_1)(A_2\lambda-\lambda\widetilde{B}_2)}_F\\
        &\leq\frac{3}{4}\norm*{(\{A_2,A_3\}-2A_1)\lambda}_F+\frac{1}{2}\norm*{\parens*{\{A_2,A_1\}-2A_3}\lambda}_F+\norm*{A_2\lambda-\lambda\widetilde{B}_2}_F\\
        &\leq (\frac{3}{4}+\frac{1}{2})2^4\sqrt{117\varepsilon}+2^3\sqrt{13\varepsilon}\\
        &=68\sqrt{13\varepsilon}\qedhere
    \end{align*}
\end{proof}

\subsubsection{Proof of the reduction}

First, we show that near-optimal strategies for the tilted bipartite cube game give rise to operators that are near-perfect for a single E3-LIN relation. This is necessary in the proof of soundness.

\begin{lemma}\label{lem:near-optimal-to-3-lin}
    Let $p(a,b|x,y)=\braket{\psi}{P^x_a\otimes Q^y_b}{\psi}$ be an $\varepsilon$-optimal strategy for $G^{\tiltedsquare}$ such that the players' measurements are PVMs. Then, there exists a state $\ket{\phi}\in H_A\otimes H_A$ that depends only on $\ket{\psi}$, and a POVM $\{\Pi_{\mbf{a}}\}_{\mbf{a}\in\{0,1\}^3}\subseteq\mc{B}(H_A)$ such that
    \begin{align*}
        \frac{1}{3}\sum_{\substack{\mbf{a}\in\{0,1\}^3\\a_1+a_2+a_3=0}}\braket{\phi}{\Pi_{\mbf{a}}\otimes(P^{011}_{a_1}+P^{101}_{a_2}+P^{110}_{a_3})}{\phi}\geq 1-C\varepsilon,
    \end{align*}
    where $C=301814$.
\end{lemma}

\begin{proof}
    We use the same notation as in \cref{lem:near-optimal}, and note that Alice and Bob's observables are $A_x=P^x_0-P^x_1$ and $B_y=Q^y_0-Q^y_1$, respectively. Fix an orthonormal basis $\{\ket{k}\}_k$ of $H_A$, and let $\ket{\phi}=\sum_k\ket{k}\otimes\psi_A^{1/2}\ket{k}$. Let $\Pi_{\mbf{a}}=\overline{P^3_{a_3}P^2_{a_2}P^1_{a_1}P^2_{a_2}P^3_{a_3}}$. Write
    \begin{align*}
        w=\frac{1}{3}\sum_{\substack{\mbf{a}\in\{0,1\}^3\\a_1+a_2+a_3=0}}\braket{\phi}{\Pi_{\mbf{a}}\otimes(P^{1}_{a_1}+P^{2}_{a_2}+P^{3}_{a_3})}{\phi}.
    \end{align*}

    We will lower bound this by upper bounding
    \begin{align*}
        1-w=\frac{1}{3}\sum_{\substack{\mbf{a}\in\{0,1\}^3\\a_1+a_2+a_3=0}}\braket{\phi}{\Pi_{\mbf{a}}\otimes(P^{1}_{\lnot a_1}+P^{2}_{\lnot a_2}+P^{3}_{\lnot a_3})}{\phi}+\sum_{\substack{\mbf{a}\in\{0,1\}^3\\a_1+a_2+a_3=1}}\braket{\phi}{\Pi_{\mbf{a}}\otimes I}{\phi}.
    \end{align*}

    First, we want to show that $A_i$ approximately commutes with $\psi_A^{1/2}$. Due to \cref{lem:symmetric-xor}, $\norm{A_i\lambda-\lambda\widetilde{B}_i}_F\leq 8\sqrt{13\varepsilon}$. With the goal of removing $\widetilde{B}_i$, we first round it to an order-$2$ unitary. We have
    \begin{align*}
        \norm{\lambda(I-\widetilde{B}_i^2)}_F\leq \norm{A_i^2\lambda-A_i\lambda\widetilde{B}_i}_F+\norm{A_i\lambda\widetilde{B}_i-\lambda\widetilde{B}_i^2}_F\leq3\norm{A_i\lambda-\lambda\widetilde{B}_i}_F\leq 24\sqrt{13\varepsilon}.
    \end{align*}
    Now, let $C_i=\mathrm{sgn}(\widetilde{B}_i)$, which is an order-$2$ unitary by construction. We have that
    \begin{align*}
        (I-\widetilde{B}_i^2)^2=(C_i-\widetilde{B}_i)^2(C_i+\widetilde{B}_i)^2\geq(C_i-\widetilde{B}_i)^2,
    \end{align*}
    so $\norm{\lambda(C_i-\widetilde{B}_i)}_F\leq\norm{\lambda(I-\widetilde{B}_i^2)}_F\leq24\sqrt{13\varepsilon}$. Therefore, we can replace $\widetilde{B}_i$ with $C_i$ and find
    \begin{align*}
        \norm{A_i\lambda-\lambda C_i}_F\leq\norm{A_i\lambda-\lambda \widetilde{B}_i}_F+\norm{\lambda(C_i-\widetilde{B}_i)}_F\leq32\sqrt{13\varepsilon}.
    \end{align*}
    To finish this step, we employ the Araki-Yamagami inequality~\cite{AY81}, which gives a tight modulus of continuity for the operator absolute value with respect to the Frobenius norm:
    \begin{align*}
        \norm*{|S|-|T|}_F\leq\sqrt{2}\norm*{S-T}_F.
    \end{align*}
    Since $|\lambda^\ast A_i|=\sqrt{A_i\lambda\lambda^\ast A_i}=A_i\psi_A^{1/2}A_i$ and $|C_i\lambda^\ast|=\sqrt{\lambda C_i^2\lambda^\ast}=\psi_{A}^{1/2}$,
    \begin{align*}
        \norm*{[A_i,\psi_A^{1/2}]}_F&=\norm*{A_i\psi_A^{1/2}A_i-\psi_A^{1/2}}_F=\norm{|\lambda^\ast A_i|-|C_i\lambda^\ast|}_F\leq\sqrt{2}\norm{\lambda^\ast A_i-C_i\lambda^\ast}_F\leq32\sqrt{26\varepsilon}.
    \end{align*}

    Next, we show that the $A_i$ approximately commute. In fact, using the result of \cref{lem:near-optimal},
    \begin{align*}
        \norm{[A_1,A_2]\psi_A^{1/2}}_F&=\norm{(A_1A_2-A_2A_1)\lambda}_F\\
        &\leq\norm{(A_1A_2-A_3)\lambda}_F+\norm{(A_2A_1-A_3)\lambda}_F\\
        &=\norm{(A_3A_1A_2-I)\lambda}_F+\norm{(A_3A_2A_1-I)\lambda}_F\\
        &\leq136\sqrt{13\varepsilon}.
    \end{align*}

    By symmetry, this bound holds for the other two commutators as well. These commutation relations also apply to the PVM elements: $\norm{[P^i_a,\psi_A^{1/2}]}_F\leq 16\sqrt{26\varepsilon}$ and $\norm{[P^i_a,P^j_b]\psi_A^{1/2}}_F\leq34\sqrt{13\varepsilon}$. Now, we bound the terms of $1-w$ one by one. First, note that, for any $i$,
    \begin{align*}
        \braket{\phi}{\Pi_{\mbf{a}}\otimes P^i_{\lnot a_i}}{\phi}&=\Tr\squ*{\overline{\Pi}_{\mbf{a}}\psi_A^{1/2}P^i_{\lnot a_i}\psi_A^{1/2}}=\norm*{P^1_{a_1}P^2_{a_2}P^3_{a_3}\psi_A^{1/2}P^i_{\lnot a_i}}_F^2.
    \end{align*}
    If $i=3$,
    \begin{align*}
        \norm*{P^1_{a_1}P^2_{a_2}P^3_{a_3}\psi_A^{1/2}P^3_{\lnot a_3}}_F&\leq \norm*{P^1_{a_1}P^2_{a_2}P^3_{a_3}P^3_{\lnot a_3}\psi_A^{1/2}}_F+\norm*{[P^3_{a_3},\psi_A^{1/2}]}_F\\
        &\leq 16\sqrt{26\varepsilon};
    \end{align*}
    if $i=2$,
    \begin{align*}
        \norm*{P^1_{a_1}P^2_{a_2}P^3_{a_3}\psi_A^{1/2}P^2_{\lnot a_2}}_F&\leq\norm*{P^1_{a_1}P^3_{a_3}P^2_{a_2}\psi_A^{1/2}P^2_{\lnot a_2}}_F+\norm*{[P^2_{a_2},P^3_{a_3}]\psi_A^{1/2}}_F\\
        &\leq\norm*{P^1_{a_1}P^3_{a_3}P^2_{a_2}P^2_{\lnot a_2}\psi_A^{1/2}}_F+\norm*{[P^2_{a_2},\psi_A^{1/2}]}_F+\norm*{[P^2_{a_2},P^3_{a_3}]\psi_A^{1/2}}_F\\
        &\leq 16\sqrt{26\varepsilon}+34\sqrt{13\varepsilon}=2(8\sqrt{2}+17)\sqrt{13\varepsilon};
    \end{align*}
    and if $i=1$,
    \begin{align*}
        \norm*{P^1_{a_1}P^2_{a_2}P^3_{a_3}\psi_A^{1/2}P^1_{\lnot a_1}}_F&\leq\norm*{P^1_{a_1}P^2_{a_2}\psi_A^{1/2}P^3_{a_3}P^1_{\lnot a_1}}_F+16\sqrt{26\varepsilon}\\
        &\leq\norm*{P^1_{a_1}P^2_{a_2}\psi_A^{1/2}P^1_{\lnot a_1}P^3_{a_3}}_F+16\sqrt{26\varepsilon}+34\sqrt{13\varepsilon}\\
        &\leq\norm*{P^2_{a_2}P^1_{a_1}\psi_A^{1/2}P^1_{\lnot a_1}P^3_{a_3}}_F+16\sqrt{26\varepsilon}+68\sqrt{13\varepsilon}\\
        &\leq\norm*{P^2_{a_2}P^1_{a_1}P^1_{\lnot a_1}\psi_A^{1/2}P^3_{a_3}}_F+32\sqrt{26\varepsilon}+68\sqrt{13\varepsilon}\\
        &=4(8\sqrt{2}+17)\sqrt{13\varepsilon}.
    \end{align*}
    For the final term,
    \begin{align*}
        &\sum_{\substack{\mbf{a}\in\{0,1\}\\a_1+a_2+a_3=1}}\braket{\phi}{\Pi_{\mbf{a}}\otimes I}{\phi}=\sum_{\substack{\mbf{a}\in\{0,1\}\\a_1+a_2+a_3=1}}\Tr\parens*{P^3_{a_3}P^2_{a_2}P^1_{a_1}P^2_{a_2}P^3_{a_3}\psi_A}\\
        &=\frac{1}{32}\hspace{-5mm}\sum_{\substack{\mbf{a}\in\{0,1\}\\a_1+a_2+a_3=1}}\hspace{-7mm}\Tr\squ*{(I+(-1)^{a_3}A_3)(I+(-1)^{a_2}A_2)(I+(-1)^{a_1}A_1)(I+(-1)^{a_2}A_2)(I+(-1)^{a_3}A_3)\psi_A}\\
        &=\frac{1}{32}\Tr\squ*{\parens*{4I-4A_1A_2A_3+4I-4A_2A_1A_3+4I-4A_3A_1A_2+4I-A_3A_2A_1}\psi_A}\\
        &=\frac{1}{8}\parens*{\norm*{(A_1A_2A_3-I)\lambda}_F^2+\norm*{(A_2A_1A_3+1)\lambda}_F^2}\\
        &\leq\frac{1}{4}(68\sqrt{13\varepsilon})^2=4\cdot13\cdot17^2\varepsilon=15028\varepsilon.
    \end{align*}
    Putting everything together, we find the bound
    \begin{align*}
        1-w&\leq\frac{4}{3}\parens*{(16\sqrt{26\varepsilon})^2+(2(8\sqrt{2}+17)\sqrt{13\varepsilon})^2+(4(8\sqrt{2}+17)\sqrt{13\varepsilon})^2}+15028\varepsilon\\
        &=\frac{4}{3}\parens*{6656\varepsilon+260(417+272\sqrt{2})\varepsilon}+15028\varepsilon.\\
        &\leq 301814\varepsilon\qedhere
    \end{align*}
    
\end{proof}

Now, we show the converse statement: operators satisfying a single E3-LIN relation near-perfectly induces a near-optimal strategy for the tilted bipartite cube game. This is necessary for the proof of completeness.

\begin{lemma}\label{lem:completeness}
    Let $\ket{\psi}\in H_A\otimes H_B$ be a quantum state, and let $\{P^{u}_a\}_{a\in\{0,1\}}\subseteq\mc{B}(H_B)$ for $u\in\{011,101,110\}$and $\{\Pi_{\mbf{a}}\}_{\mbf{a}\in\{0,1\}^3}\subseteq\mc{B}(H_A)$ be PVMs satisfying
    \begin{align*}
        \frac{1}{3}\sum_{\substack{\mbf{a}\in\{0,1\}^3\\a_1+a_2+a_3=0}}\braket{\psi}{\Pi_{\mbf{a}}\otimes(P^{011}_{a_1}+P^{101}_{a_2}+P^{110}_{a_3})}{\psi}\geq 1-\varepsilon.
    \end{align*}
    Then, there exist PVMs $\{Q^u_a\}_{a\in\{0,1\}}\subseteq\mc{B}(H_A)$ for $u\in\{001,010,100,111\}$ such that $p(a,b|x,y)=\braket{\psi}{Q^y_b\otimes P^x_a}{\psi}$ is a $\frac{3}{8}\varepsilon$-optimal strategy for the tilted bipartite cube game.
    
\end{lemma}

\begin{proof}
    Without loss of generality we may suppose that $\Pi_{\mbf{a}}=0$ if $a_1+a_2+a_3=1$ --- we can change Alice's PVM to one satisfying this condition without decreasing the value. Now write $\Pi^{011}_a=\sum_{a_2,a_3}\Pi_{aa_2a_3}$, $\Pi^{101}_a=\sum_{a_1,a_3}\Pi_{a_1aa_3}$, and $\Pi^{110}_a=\sum_{a_1,a_2}\Pi_{a_1a_2a}$. Let $A_u=P^u_0-P^u_1$ and $C_u=\Pi^u_0-\Pi^u_1$ for all $u\in\{011,101,110\}$; let $A_{000}=I$ and $C_{000}=C_{011}C_{101}C_{110}$; and define $B_v=\frac{1}{2}\parens*{C_{\lnot v}-C_{v+e_1}-C_{v+e_2}-C_{v+e_3}}$ for all $v\in X_1$. Since the $C_u$ are commuting order-$2$ unitaries satisfying $C_{000}C_{011}C_{101}C_{110}=I$, the $B_u$ are also order-$2$ unitaries.
    Next, writing $C$ for the correlation corresponding to the strategy $p$, note that
    \begin{align*}
        1-\varepsilon&=\frac{1}{3}\sum_{\substack{u\in\{011,101,110\}\\a\in\{0,1\}}}\braket{\psi}{\Pi^u_a\otimes P^u_{a}}{\psi}\\
        &=\frac{1}{6}\sum_{u\in X_0}\parens*{1+\braket{\psi}{C_u\otimes A_u}{\psi}}-\frac{1}{6}\parens*{1+\braket{\psi}{C_{000}\otimes I}{\psi}}\\
        &=\frac{1}{2}+\frac{1}{6}\sum_{u\in X_0}\braket{\psi}{\frac{1}{2}\parens*{B_{\lnot u}-B_{u+e_1}-B_{u+e_2}-B_{u+e_3}}\otimes A_u}{\psi}\\
        &\qquad-\frac{1}{6}\sum_{\mbf{a}\in\{0,1\}^3}(-1)^{a_1+a_2+a_3}\braket{\psi}{\Pi_{\mbf{a}}\otimes I}{\psi}\\
        &=\frac{1}{2}+\frac{4}{3}\beta(G^{\tiltedsquare},C)-\frac{1}{6}\braket{\psi}{I\otimes I}{\psi}=\frac{1}{3}+\frac{4}{3}\beta(G^{\tiltedsquare},C).
    \end{align*}
    Hence, we can rearrange and find that $\beta(G^{\tiltedsquare},C)\geq\frac{1}{2}-\frac{3}{4}\varepsilon$, as wanted.

\end{proof}

Now we can pass to the proof of the main theorem.

\begin{proof}[Proof of \cref{thm:main-robust}]
    Let $(S,\pi)$ be an instance of succinct E3-LIN. We construct an instance $G$ of succinct TXOR as follows. We replace each constraint with a copy of the tilted bipartite cube game, where we identify the questions $011$, $101$, and $110$ with the original variables in the constraint, and identify the $000$ variables from all the cubes, and take this to be the distinguished variable. For those constraints of the form $x_1+x_2+x_3=1$, we flip the signs of the edges adjacent to $000$ in the corresponding cube. This reduction is local on each constraint, so it can be effected in polynomial-time on a succinctly-presented instance.

    First, we show completeness. If there exists a quantum strategy $p$ such that $\omega(G(S,\pi),p)\geq 1-\varepsilon$. Using Naimark dilation, we may assume it is projective. Let $1-\varepsilon_{b,\mbf{x}}$ be the value on the constraint $x_1+x_2+x_3=b$. Then, by constructing a strategy $p'$ for $G$ using \cref{lem:completeness}, the value of the corresponding cube is $\geq \frac{3}{4}-\frac{3}{8}\varepsilon_{b,\mbf{x}}$, and therefore $\omega(G,p')\geq \frac{3}{4}-\frac{3}{8}\varepsilon$.

    Next, we show soundness. Suppose there exists a quantum strategy $p$ for $G$ such that $\omega(G,p)\geq\frac{3}{4}-\delta$. As above, we can suppose the strategy is projective due to Naimark dilation. Then, via \cref{lem:near-optimal-to-3-lin}, there exists a strategy $p'$ for $G(S,\pi)$ such that $\omega(G(S,\pi),p')\geq 1-4\cdot10^6\delta$.

    Putting these together, the construction maps $\sucethreelin^\ast_{1-\varepsilon,s}$ to $\suctxor^\ast_{\frac{3}{4}-\frac{3}{8}\varepsilon,\frac{3}{4}-\frac{1-s}{4\cdot 10^5}}$.
\end{proof}

\section{Hardness as a function of completeness and soundness values}\label{sec:hardness-for-c-and-s}

In this section, we show that the best known techniques for approximating the classical value of an XOR game based on the quantum value, due to~\cite{CHTW04}, extend approximation algorithms for the quantum value of tilted XOR games. We also study the range of values where \cref{thm:main-robust} implies approximating the quantum value of a tilted XOR game is hard.

The results proved in this section correspond to the approximation regions shown in \cref{fig:approx-prelim}.

\begin{lemma}\label{lem:lowerbound}
    Let $G$ be a tilted XOR game. Then, there exists a classical strategy $p$ for $G$ such that $\omega(G,p)\geq\frac{1}{2}$. 
\end{lemma}

This also implies that $\txor_{c,s}$ and $\txor_{c,s}^\ast$ with $c\leq\frac{1}{2}$ are trivial, as every instance is a yes instance.

\begin{proof}
    We define a classical strategy by sampling a random deterministic strategy $p(a,b|x,y)=\delta_{a,g(x)}\delta_{b,h(y)}$ for $g$ and $h$ uniformly random functions. Then, the bias is $\geq0$, giving that the value of the strategy is $\geq\frac{1}{2}$.
\end{proof}

\begin{lemma}\label{lem:how-easy}
    Let $\gamma:[\frac{1}{2},1]\rightarrow[\frac{1}{2},1]$ be a monotone increasing function such that for all XOR games~$G$, it holds that $\omega^\ast(G)\leq\gamma(\omega(G))$. Then, if $c> \gamma(s)$, there exist polynomial-time algorithms to decide $\txor_{c,s}^\ast$ and $\txor_{c,s}$.
\end{lemma}

\begin{proof}
    Consider the following algorithm (the input is a tilted XOR game $G$):
    \begin{enumerate}[1.]
        \item Compute $\omega^\ast(G_{\xor})$ to precision $\frac{1}{2}(c-\gamma(s))$, where $G_{\xor}$ is the XOR relaxation of $G$.
        \item If $\omega^\ast(G_{\xor})\geq c$ then output YES; otherwise, output NO.
    \end{enumerate}
    Since $\omega^\ast(G_{\xor})$ can be computed by an SDP, this algorithm is polynomial-time in the instance size and $\log(1/(c-\gamma(s)))$. Now, if $\omega^\ast(G)\geq c$, then $\omega^\ast(G_{\xor})\geq\omega^\ast(G)\geq c$, so the algorithm outputs YES correctly. On the other hand, if $\omega^\ast(G)<s$, then, $\omega(G_{\xor})=\omega(G)<s$, so using the assumption on $\gamma$, $\omega^\ast(G_{\xor})<\gamma(s)< c$. As such, the algorithm outputs NO correctly.

    The proof for the classical case proceeds identically.
\end{proof}

The next theorem provides an explicit form of the function $\gamma$ needed above.

\begin{theorem}[\cite{CHTW04}]\label{thm:from-first}
    There exists a monotone increasing function $\gamma:[\frac{1}{2},1]\rightarrow[\frac{1}{2},1]$ such that $\omega^\ast(G)\leq\gamma(\omega(G))$ for all XOR games $G$, where
    \begin{align*}
        \gamma(x)=\begin{cases}\frac{1}{2}+\kappa\parens*{x-\frac{1}{2}}&\frac{1}{2}\leq x\leq\gamma_0\\\gamma_1x&\gamma_0<x\leq\gamma_2\\\sin^2\parens*{\frac{\pi}{2}x}&\gamma_2<x\leq1,\end{cases}
    \end{align*}
    where $\kappa\approx 1.7822$ is an upper bound on Grothendieck's constant, $\gamma_0=\frac{\kappa-1}{2(\kappa-\gamma_1)}\approx0.60730$ is the point where the two line segments intersect, and $\gamma_1\approx 1.1382$ and $\gamma_2\approx0.74202$ are such that such that $\gamma_1x$ is tangent to $\sin^2\parens*{\frac{\pi}{2}x}$ at $0<\gamma_2<1$.
\end{theorem}

\begin{lemma}\label{lem:hard-convex}
    For any $c\geq s$, the decision problem $\txor^\ast_{c,s}$ reduces to $\txor^\ast_{c',s'}$ for any $(c',s')$ that is a convex combination of $(c,s)$, $(1,1)$ and $(\frac{1}{2},\frac{1}{2})$, as long as the weight of $(c,s)$ is nonzero. The same holds for the classical problems.
\end{lemma}

\begin{proof}
    The first reduction follows simply by taking a convex combination with the XOR games with winning probability $1$ and $\frac{1}{2}$.
\end{proof}

\section{Extensions to other models}\label{sec:other-models}

In this section, we extend our hardness of approximation result for tilted XOR games to modified models. We consider modified games, or modified models of entanglement, and discuss connections between our results and other work.

\subsection{Tilted XOR games with only one instance of a non-XOR  question}\label{sec:single-bit-type}

Recall that, for the tilted cube game (\cref{sec:tilted-cube}), twelve of the sixteen question instances are XOR and four are non-XOR (where the winning condition depends only on Bob's answer bit).
Therefore, our results about the hardness of tilted XOR games in \cref{sec:hardness-of-tilted-XOR} are for tilted XOR games with many different question instances that are non-XOR.
In this section, we show that the \tsf{RE}-hardness of tilted XOR games may be preserved even when the difference with XOR games is minimised to a single non-XOR question instance.

\begin{theorem}\label{thm:one-bit-type-question}
    There exist $c>s$ such that it is $\tsf{RE}$-hard to decide if the quantum value is $\geq c$ or~$<s$, for games that take the form of an XOR game, with one additional question asked to Alice where her answer is correct if and only if it is $0$.
\end{theorem}

This corresponds to tilted XOR games where the distinguished variable is only asked in one question instance.
Note that, for reasons of convenience, we state and prove \cref{thm:one-bit-type-question} in the symmetric case where the distinguished variable $\perp$ is asked to Bob (rather than Alice).

\begin{proof}
    Let $x$ be a Turing machine, and let $G_x$ be the corresponding tilted XOR game due to \cref{thm:main-robust}. Recall that this reduction gives completeness $\frac{3}{4}-\varepsilon$ and soundness $\frac{3}{4}-\delta$ for some constant $\delta>0$ and $\varepsilon>0$ that may be chosen arbitrarily small. Now, consider the game $G'_x$ where, with probability $\frac{1}{2}$, $(G_x)_{\xor}$ is played and, with probability $\frac{1}{2}$, $G_\perp$ is played, which is the game where Alice is asked $\perp$ and wins iff she responds with $0$. By using the optimal strategy for $G_x$, we have that $\omega^\ast(G_x')\geq \frac{1}{2}\omega^\ast(G_x)+\frac{1}{2}$. Thus, we take $c=\frac{1}{2}\parens*{\frac{3}{4}-\varepsilon}+\frac{1}{2}=\frac{7}{8}-\varepsilon$ to get completeness. Now we show soundness. Fix $s=\frac{7}{8}-\eta$ with $\eta$ to be specified later. If we suppose that for some strategy~$p$, $\omega(G_x',p)\geq s$, we have that
    \begin{align*}
        \omega(G_\perp,p)\geq2\parens*{s-\frac{1}{2}\omega\bigl((G_x)_{\xor},p\bigr)}\geq2\parens*{s-\frac{3}{8}}=1-2\eta.
    \end{align*}
    But the value of $G_\perp$ under strategy $p$ is simply $\Tr(P^{\perp}_0\rho_A)=\frac{1}{2}+\frac{1}{2}\Tr(A_\perp\rho_A)$. Hence,
    \begin{align*}
        \norm[\big]{((A_\perp-I)\otimes I)\ket{\psi}}^2=4\braket{\psi}{P^\perp_1\otimes I}{\psi}\leq 8\eta.
    \end{align*}
    Now consider the strategy $p'$ for $G_x$ with the same observables as $p$ --- this corresponds to the strategy for $(G_x)_{\xor}$ with $A_\perp$ replaced by $I$. We have that
    \begin{align*}
        \omega(G_x,p') \geq\omega\bigl((G_x)_{\xor},p\bigr)-\norm[\big]{(A_\perp\otimes I)\ket{\psi}}
        \geq\frac{3}{4}-2\eta-2\sqrt{2\eta}.
    \end{align*}
    Let $\eta=\frac{\delta^2}{16}$. Then, $\omega(G_x,S')\geq\frac{3}{4}-\frac{\delta^2}{8}-\frac{\delta}{\sqrt{2}}\geq\frac{3}{4}-\delta$, giving soundness for the reduction as well.
\end{proof}

\subsection{Commuting-operator strategies for tilted XOR games}\label{sec:commuting-operator}

In this section, we consider strategies for tilted XOR games in the commuting-operator model. We prove that the gapless problem of deciding the commuting-operator value of a tilted XOR game is $\tsf{coRE}$-complete, an analogue of \cref{thm:main-exact}.

\begin{definition}
    Let $G$ be a nonlocal game. 
    A strategy $p$ for $G$ is \emph{commuting-operator} if there exists a (possibly infinite-dimensional) Hilbert space $H$, POVMs $\{P^x_a\}_{a\in A}$ and $\{Q^y_b\}_{b\in B}$ in $\mc{B}(H)$ for every $x\in X$ and $y\in Y$, and a state $\ket{\psi}\in H$, such that $[P^x_a,Q^y_b]=0$ for every $(a,b,x,y)\in A\times B\times X\times Y$, and $p(a,b|x,y)=\braket{\psi}{P^x_aQ^y_b}{\psi}$.

    The \emph{commuting-operator value} of $G$ is the supremum over the values of all commuting-operator strategies; it is denoted $\omegaco(G)$.
\end{definition}

Similarly to quantum strategies, the POVMs in a commuting operator strategy can always chosen to be PVMs. See, for example, Lemma 3.4(b) in \cite{Fri12}.

For $\frac{1}{2}\leq s\leq c\leq1$, we define $\txor_{c,s}^\co$ as the problem of deciding, for a given tilted $\xor$ game $G$, if $\omegaco(G)\geq c$ or $\omegaco(G)<s$, with the promise that one of these two holds. Similarly, we define $\ethreelin_{c,s}^\co$ as the problem of deciding, for a given $\ethreelin$ game $G$, if $\omegaco(G)\geq c$ or $\omegaco(G)<s$, with the promise that one of these two holds.

Given a commuting-operator strategy $p$ for a XOR or tilted XOR game $G$, we can associate it with the commuting-operator correlation $C(x,y)=\braket{\psi}{A_xB_y}{\psi}$, where $A_x=P^x_0-P_1^x$ and $B_y=Q^y_0-Q^y_1$.

Before continuing to the proof of the main theorem of this section, we need the following lemma, which is a commuting-operator version of \cref{lem:optimal-case}.

\begin{lemma}\label{lem:exact_co}
    Let $C(x,y)= \braket{\psi}{A_xB_y}{\psi}$ be an optimal commuting-operator correlation for the tilted cube game $G^{\tiltedsquare}$.
    Suppose that the observables are unitaries.
    Then, for every $x,y\in \{011,101,110\}$, we have $[A_x,A_y]\ket{\psi}=0$ and $A_{011}A_{101}A_{110}\ket{\psi}=\ket{\psi}$.
\end{lemma}
\begin{proof}
    The proof is the same proof as of \cref{lem:optimal-case}, but where we keep all the algebraic equalities up to the action of the involving operators on the state $\ket{\psi}$.
\end{proof}

\begin{theorem}\label{thm:commuting-operator}
    The problem $\txor_{\scriptscriptstyle \frac{3}{4},\frac{3}{4}}^\co$ is \tsf{coRE}-complete.
\end{theorem}
\begin{proof}
     It is known that $\ethreelin_{1,1}^\co$ is \tsf{coRE}-hard \cite{Slo20}.
    We shall prove that $\ethreelin_{1,1}^\co$ reduces to $\txor_{\scriptscriptstyle \frac{3}{4},\frac{3}{4}}^\co$.
    The proof is similar to the proof of \cref{thm:main-exact}.

    Let $G$ and $\widetilde{G}$ be as in the proof of \cref{thm:main-exact}.
    We start with completeness.
    Assume $G$ has commuting-operator value 1. Then, by \cite{CLS17}, there are a Hilbert space $H'$, a state $\ket{\psi}\in H'$, and binary observables $A_i$ and $B_j$ such that $A_i$ and $B_j$ commute for all $i,j$; for every variable $v_i$, we have $A_i\ket{\psi}=B_i\ket{\psi}$; and for every equation $v_i\oplus v_j\oplus v_k=b$, we have that $[A_x,A_y]\ket{\psi}=0$ for every $x,y\in\{i,j,k\}$ and $A_iA_jA_k\ket{\psi}=(-1)^b\ket{\psi}$.
    Denote $\mathcal{A}$ and $\mathcal{B}$ the unital algebras generated by Alice and Bob observables, respectively.
    As in the proof of Lemma 8 in \cite{CLS17}, we set 
    $H_0 = \overline{\mathcal{A}\ket{\psi}}=\overline{\mathcal{B}\ket{\psi}}$, 
    and for every $A,A'\in \mathcal{A}$, we have $A\ket{\psi}=A'\ket{\psi}$ if and only if $A|_{H_0}=A'|_{H_0}.$
    A similar relation holds for the operators in $\mathcal{B}$.
    
    We will now produce a strategy for the tilted XOR game that succeeds with probability $\frac{3}{4}$. 
    Let the Hilbert space of the players be $H = H_0$ and let their shared state be $\ket{\psi}$.
    The following operators are defined up to their action on $H$. 
    For each cube within the tilted XOR game $\widetilde{G}$ (as illustrated in \cref{fig:tilted-cube-game-reduction}), assign the operators $A_i,A_j,$ and $A_k$ to the vertices labelled $i$, $j$, and $k$, and assign the operator $I$ to the vertex labeled $\perp$. For the remaining (odd parity) vertices of the cube, assign (separately for each cube) the binary observables
\medskip
\begin{align*}
    \widetilde{B}_{111} &= \textstyle{\frac{1}{2}}(+(-1)^bI - B_i - B_j - B_k) \\[1mm]
    \widetilde{B}_{100} &= \textstyle{\frac{1}{2}}(-(-1)^bI + B_i - B_j - B_k) \\[1mm]
    \widetilde{B}_{010} &= \textstyle{\frac{1}{2}}(-(-1)^bI - B_i + B_j - B_k) \\[1mm]
    \widetilde{B}_{001} &= \textstyle{\frac{1}{2}}(-(-1)^bI - B_i - B_j + B_k).
\end{align*}
Then it is a straightforward exercise to show that the above are proper binary observables and 
that the success probability of this strategy for the tilted XOR game $\widetilde{G}$ is $\frac{3}{4}$.

In the other direction, let $\ket{\psi}, A_x,B_y$ be a commuting-operator strategy for the tilted XOR game~$\widetilde{G}$ that wins with probability $\frac{3}{4}$.
As in the proof of \cref{thm:main-exact}, it follows that the strategy succeeds in each of the tilted cube games with probability $\frac{3}{4}$.
It follows from \cref{lem:exact_co}, that for each equation of the form $v_i\oplus v_j\oplus v_k=b$, the corresponding observables $A_i,A_j$ and $A_k$ commute (up to their action on $\ket{\psi}$) and $A_iA_jA_k\ket{\psi}=(-1)^b\ket{\psi}$.
Let $B_{111},B_{100},B_{010}$ and $B_{001}$ the corresponding Bob's observables for the odd parity vertices (on the cube associated to the given equation).
We also have, by \cref{thm:xor-rigidity-exact},
\begin{align*}
        \sum_{y}H_{x,y}B_y\ket{\psi}&=\frac{1}{8}A_x\ket{\psi},\\
        \sum_{x}H_{x,y}A_x\ket{\psi}&=\frac{1}{8}B_y\ket{\psi}.
    \end{align*}
Denote $\mathcal{A}$ and $\mathcal{B}$ the unital algebras generated by Alice and Bob observables, respectively.
By similar arguments as in the proof of Lemma 8 in \cite{CLS17}, we set 
$H_0 =\overline{\mathcal{A}\ket{\psi}}=\overline{\mathcal{B}\ket{\psi}}$, and for every $A,A'\in \mathcal{A}$, it holds that $A\ket{\psi}=A'\ket{\psi}$ if and only if $A|_{H_0}=A'|_{H_0}.$
Thus, the operators $A_i|_{H_0}$, $A_j|_{H_0}$, and $A_k|_{H_0}$ commute and $A_iA_jA_k|_{H_0}=(-1)^bI_{H_0}$.
Since the choice of $H_0$ is independent of the equation, we get that the collection of all operators $A_i|_{H_0}$ is an operator solution for $G$; thus $\omegaco(G)=1$, as required.
\end{proof}

\subsection{Oracularised games and oracularisable strategies}

In the \emph{oracularisation} of a nonlocal game, as defined in~\cite{NW19}, Alice is asked both her and Bob's questions and Bob is asked one of the two questions. The players win if Alice answers both answers correctly and Bob's answer is consistent. The oracularisation transformation is sound in the sense that if the value of a game is $\leq 1-\varepsilon$, then the value of its oracularisation is $\leq 1-\mathrm{poly}(\varepsilon)$. On the other hand, the oracularisation is not necessarily complete for perfect strategies~\cite{NW19}. A game is called \emph{oracularisable} if its oracularisation admits a perfect quantum strategy. This condition is actually equivalent to an algebraic condition on the perfect quantum strategies: a game is oracularisable if and only if it admits a perfect quantum tracial strategy where the two observables for any question pair commute (note that this condition is not vacuous as the observables for a tracial strategy are seen as acting on the same space rather than two spaces corresponding to the two players).
Strategies of this form, even if they are not perfect, were referred to as commuting and consistent strategies in~\cite{JNV+20arxiv}, but have since become known as \emph{oracularisable strategies}~\cite{MS24,MS25,CM24,TV25,Lin25}.

In this section, we study oracularisation and oracularisable strategies in the context of tilted XOR games. Oracularised tilted XOR games correspond to the constraint-variable games, as in \cite{CM14a,CM24}, of linear systems with one or two variables per equation; and oracularisable strategies were studied in the context of the quantum unique games conjecture~\cite{MS25}. We show first that our reduction does not extend to oracularised tilted XOR games, leaving the complexity of these games open; but also show that both tilted XOR and XOR games are hard when the strategies are restricted to oracularisable strategies.

\subsubsection{Obstacle for oracularised version of tilted XOR games}\label{sec:oracularised}

Linear systems are naturally presented as oracular games, where one of the players assigns values to a full constraint and the other assigns a value to one of the variables of the constraint (this is called the constraint-variable setting in, \textit{e.g.}, \cite{CM24}). However, in an XOR or tilted XOR game, each of the players is asked for an assignment to a single variable. Since these games correspond to linear systems with one or two variables per equation, it is natural to ask about their behaviour in the oracularised setting. However, the reduction in this work does not seem to extend to oracular games. We show in this section that quantum values of the bipartite cube game and its tilted analogue differ in the oracularised setting, and hence we cannot use the XOR game structure theorem to understand the behaviour of almost-perfect strategies for the oracularised tilted bipartite cube game.

\begin{definition}
    Let $G=(X,Y,A,B,\pi,V)$ be a nonlocal game. The (projection) \emph{oracularisation} of $G$ is the nonlocal game $G^{\orac}=(X\times Y,X\sqcup Y,A\times B,A\cup B,\pi^{\orac},V^{\orac})$, where
    $$\pi^{\orac}((x,y),z)=\begin{cases}\frac{1}{2}\pi(x,y)&z=x\lor z=y\\0&\text{else,}\end{cases}$$
    and $$V^{\orac}((a,b),c|(x,y),z)=\begin{cases}1&V(a,b|x,y)=1\land z=x\land c=a\\1&V(a,b|x,y)=1\land z=y\land c=b\\0&\text{else.}\end{cases}$$
\end{definition}

First, we note that the oracularisation of an XOR can still be expressed as an XOR game.

\begin{lemma}\label{lem:xor-orac}
    Let $G$ be an XOR game. Then, there is an XOR game $G'$ such that $\omega(G^{\orac})=\omega(G')$ and $\omega^\ast(G^{\orac})=\omega^\ast(G')$. We have that $\omega(G^{\orac})=\frac{1+\omega(G)}{2}$ and $\frac{1+\omega^\ast(G)}{2}\leq\omega^\ast(G^{\orac})\leq \frac{1+\sqrt{\omega^\ast(G)}}{2}$. The upper bound is an equality if and only if there exists an optimal correlation $C$ such that $(-1)^{f(x,y)}C(x,y)=\beta^\ast(G)$ for all $x\in X$, $y\in Y$ such that $\pi(x,y)>0$.
\end{lemma}

\begin{proof}
    Let $G'=(X\times Y, X\sqcup Y,\{0,1\},\{0,1\},\pi^{\orac},V')$, where $V'(a,b|(x,y),z)=\delta_{a+b,f'(x,y,z)}$ for
    $$f'(x,y,z)=\begin{cases}f(x,y)&z=y\\0&\text{else.}\end{cases}$$
    By construction, $G'$ is an XOR game. Now, suppose $p$ is a strategy for $G^{\orac}$. We may suppose without loss of generality that $p((a,b),c|(x,y),z)=0$ if $a+b\neq f(x,y)$. Then, defining $p'(a,b|(x,y),z)=p((a,f(x,y)+a),b|(x,y),z)$,
    $p'$ is a strategy for $G'$ and
    \begin{align*}
        \omega(G',p')&=\sum_{x\in X,y\in Y,a\in\{0,1\}}\frac{1}{2}\pi(x,y)\parens*{p'(a,a|(x,y),x)+p'(a,a+f(x,y)|(x,y),y)}\\
        &=\sum_{x\in X,y\in Y}\sum_{a,b\in\{0,1\}:\,a+b=f(x,y)}\frac{1}{2}\pi(x,y)\parens*{p((a,b),a|(x,y),x)+p((a,b),b|(x,y),y)}\\
        &=\omega(G^{\orac},p).
    \end{align*}
    Doing the same in the other direction and noting that the mapping $p\mapsto p'$ preserves classical and quantum strategies, we have that $\omega(G^{\orac})=\omega(G')$ and $\omega^\ast(G^{\orac})=\omega^\ast(G')$.

    Now, consider a deterministic correlation $C((x,y),z)=(-1)^{g(x,y)+h(z)}$ for $G'$. Then, the bias
    \begin{align*}
        \beta(G',C)&=\sum_{x\in X,y\in Y}\frac{1}{2}\pi(x,y)\parens*{C((x,y),x)+(-1)^{f(x,y)}C((x,y),y)}\\
        &=\sum_{x\in X,y\in Y}\frac{1}{2}\pi(x,y)(-1)^{g(x,y)}\parens*{(-1)^{h(x)}+(-1)^{f(x,y)+h(y)}}.
    \end{align*}
    This is maximised by taking a function $g(x,y)$ such that $(-1)^{g(x,y)}\parens*{(-1)^{h(x)}+(-1)^{f(x,y)+h(y)}}$ is non-negative, which happens at $g(x,y)=h(x)$. This gives
    \begin{align*}
        \beta(G',C)&=\sum_{x\in X,y\in Y}\frac{1}{2}\pi(x,y)\parens*{1+(-1)^{f(x,y)+h(x)+h(y)}}=\frac{1}{2}+\frac{1}{2}\omega(G,C_0),
    \end{align*}
    where $C_0(x,y)=(-1)^{h(x)+h(y)}$. Since $C_0$ can be taken to be any deterministic correlation we get $\beta(G')=\frac{1+\beta(G)}{2}$, as wanted.

    Next, consider a quantum (vector) correlation $C((x,y),z)=\braket{u_{x,y}}{v_z}$ for $G'$. The bias is
    \begin{align*}
        \beta(G',C)&=\sum_{x\in X,y\in Y}\frac{1}{2}\pi(x,y)\parens*{\braket{u_{x,y}}{v_x}+(-1)^{f(x,y)}\braket{u_{x,y}}{v_y}}.
    \end{align*}
    To get the lower bound, note that by taking $\ket{u_{x,y}}=\ket{v_x}$,
    we get that
    \begin{align*}
        \beta(G',C)&=\sum_{x\in X,y\in Y}\frac{1}{2}\pi(x,y)\parens*{1+(-1)^{f(x,y)}\braket{v_x}{v_y}}=\frac{1}{2}+\frac{1}{2}\beta(G,C_0),
    \end{align*}
    where $C_0(x,y)=\braket{v_x}{v_y}$. Taking $C_0$ to be an optimal quantum correlation for $G$ gives $\beta^\ast(G')\geq\frac{1}{2}+\frac{1}{2}\beta^\ast(G)$. To get the upper bound, note that $\braket{u_{x,y}}{v_x}+(-1)^{f(x,y)}\braket{u_{x,y}}{v_y}$ is maximised among unit vectors $\ket{u_{x,y}}$ by taking $\ket{u_{x,y}}$ to be the normalisation of $\ket{v_x}+(-1)^{f(x,y)}\ket{v_y}$. In this case,
    \begin{align*}
        \beta(G',C)&=\sum_{x\in X,y\in Y}\frac{1}{2}\pi(x,y)\abs*{\ket{v_x}+(-1)^{f(x,y)}\ket{v_y}}\\
        &=\sum_{x\in X,y\in Y}\frac{1}{2}\pi(x,y)\sqrt{2+2(-1)^{f(x,y)}\braket{v_x}{v_y}}\\
        &\leq\sqrt{\frac{1}{2}+\frac{1}{2}\beta(G,C_0)},
    \end{align*}
    using Jensen's inequality, with $C_0$ as before. This gives the wanted upper bound, which is attained if and only if Jensen's inequality is in fact an equality here.
\end{proof}

\begin{proposition}
    $\omega(G^{\Box,\orac})=\frac{7}{8}$ and $\omega^\ast(G^{\Box,\orac})=\frac{1}{2}+\frac{\sqrt{3}}{4}$. However, $\omega^\ast(G^{\tiltedsquare,\orac})<\frac{1}{2}+\frac{\sqrt{3}}{4}$.
\end{proposition}

This implies that the structural results for optimal and near-optimal XOR game strategies due to~\cite{Slo11} cannot be use to study the quantum strategies of $G^{\tiltedsquare,\orac}$.

\begin{proof}
    The classical value of $G^{\Box,\orac}$ follows immediately from \cref{lem:xor-orac}. For the quantum value, note that, taking $H_{x,y}=\frac{(-1)^{f(x,y)}}{16}$ to be the game matrix of $G^{\Box}$, $8H$ is orthogonal. As such, there exist unit vectors $\ket{u_x},\ket{v_y}\in\R^4$ such that $\braket{u_x}{v_y}=8H_{x,y}=\frac{(-1)^{f(x,y)}}{2}$. Hence $G^{\Box}$ satisfies the necessary and sufficient condition of \cref{lem:xor-orac} for the upper bound on the quantum value of $G^{\Box,\orac}$ to be attained.

    Note that, in the oracularisation of a tilted XOR game, the predicate corresponding to the tilted equation of the form $b=f(\perp,y)$ becomes the pair of predicates $V((a,b),c|(\perp,y),\perp)=\delta_{b,f(\perp,y)}\delta_{c,a}$ and  $V((a,b),c|(\perp,y),y)=\delta_{b,f(\perp,y)}\delta_{c,b}$.
    Thus, we may assume that upon receiving a query of the form $(\perp,y)$, Alice always outputs $(0,f(\perp,y))$ and upon receiving the query $\perp$, Bob always outputs $0$ (indeed, this will only increase their winning probability).
    Now, it is direct to see that any such quantum strategy for $G^{\tiltedsquare,\orac}$ corresponds to a quantum correlation $C((x,y),z)=\braket{\psi}{A_{x,y}\otimes B_z}{\psi}$ for ${G^{\Box}}'$ where $B_\perp=I$ and $A_{\perp,y}=I$. Suppose $C$ is an optimal quantum correlation for ${G^{\Box}}'$. Using~\cite{Slo11}, we have that
    $$\sum_{z\in X\sqcup Y}(-1)^{f'(x,y,z)}\pi^{\orac}((x,y),z)(I\otimes B_z)\ket{\psi}=a_{x,y}(A_{x,y}\otimes I)\ket{\psi}.$$
    The optimal row bias $a_{x,y}=\frac{\beta^\ast({G^{\Box}}')}{16}=\frac{\sqrt{3}}{32}$ and $\pi^{\orac}((x,y),z)$ is only nonzero if $z=x$ or $z=y$, so the above simplifies to $$(I\otimes B_x+(-1)^{f(x,y)}I\otimes B_y)\ket{\psi}=\sqrt{3}(A_{x,y}\otimes I)\ket{\psi}.$$
    Squaring both sides in the standard way gives $(B_x+(-1)^{f(x,y)}B_y)^2=3I$ on the support of $\psi_B$, which simplifies to $\{B_x,B_y\}=(-1)^{f(x,y)}I$. Since $B_\perp=I$, we have that for all $y\in X_1$ that $(-1)^{f(\perp,y)}I=\{I,B_y\}=2B_y$, so $B_y=\pm\frac{I}{2}$. Repeating this again for $x\in X_0$, we have that $(-1)^{f(x,y)}I=\{B_x,\pm\frac{I}{2}\}=\pm B_x$, so $B_x=\pm I$. Therefore $C$ is actually a classical correlation, and thus cannot attain the optimal quantum value of ${G^{\Box}}'$, leading to a contradiction.
\end{proof}

The results of this section give rise to the following open question.

\begin{question}
The presentation of an E2-LIN instance as a nonlocal game that is analogous to E3-LIN is an oracularisation of an XOR game, which as noted above is itself an XOR game.
This implies that entangled E2-LIN is easy to approximate.
However, the oracularisation of a \emph{tilted} XOR game corresponds to a system of linear equations with \emph{one or two} variables per equation, and we saw that the hardness for these games cannot be proven using our technique based on the cube gadget. Is it also $\tsf{RE}$-complete to approximate the quantum value of oracularised tilted XOR games?
\end{question}

\subsubsection{Quantum oracularisible strategies for tilted XOR games}\label{subsec:oracular_value_is_hard}

In the previous section, we showed that our hardness reduction for tilted XOR games does not extend to the oracularised setting. However, it does extend to the setting where the strategies considered are oracularisable. An oracular strategy is a tracial strategy where the players' observables for questions that are asked at the same time commute. Perfect quantum strategies for the oracularisation correspond to perfect oracularisable quantum strategies --- however, beyond that, oracularisability of the strategy cannot be verified operationally.

\begin{definition}
    Let $G$ be a nonlocal game.
A strategy $p$ for $G$ is called \textbf{quantum oracularisable} if there exists a finite dimensional Hilbert space $H$ of some dimension $d$, collections of PVMs $\{P^x_a\}_a$ and $\{Q^y_b\}_b$ acting on $H$ such that $p(a,b|x,y)=\tr(P^x_aQ^y_b)$, where $\tr$ is the normalised trace, and for every $x$ and $y$ such that $\pi(x,y)>0$, we have $[P^x_a,Q^y_b]=0$ for every $a$ and $b$.
The \textbf{quantum oracularisable value} of $G$ is the supremum over the values of all quantum oracularisable strategies; it is denoted $\omegaqo(G)$.
\end{definition}

For a (tilted) XOR game $G$, the bias of $G$ against any quantum oracularisable strategy $P^x_a$, $Q^y_b$ is characterised by the corresponding quantum correlation $C(x,y)=\tr(A_xB_y)$, as in the classical and quantum cases. 

\begin{lemma}\label{lem:xor-orac-2}
    Let $G$ be a tilted XOR game. We have $\frac{1+\omegaqo(G)}{2}\leq\omegaqo(G^{\orac})\leq\frac{1+\omega^\ast(G)}{2}$.
\end{lemma}

\begin{proof}
    Let $A,B,C$ be three binary observables, $\alpha\in\{\pm 1\}$ and suppose that $A$ commutes with $B$ and with $C$.
    Then, we have
    $$\Tr(A(B+\alpha C))\leq 1+\alpha\Tr(BC).$$

    Indeed,
    $$1+\alpha\Tr(BC)-\Tr(AB)-\alpha\Tr(AC)=\Tr((I-AB)(I-\alpha AC)),$$
    and the inequality follows as $AB$ and $\alpha AC$ are self-adjoint unitaries and the fact that for non-negative operators $P$ and $Q$, $\Tr(PQ)\geq 0$.

    Note that, in the oracularisation of a tilted XOR game, the predicate corresponding to the tilted equation of the form $b=f(\perp,y)$ becomes the pair of predicates $V((a,b),c|(\perp,y),\perp)=\delta_{b,f(\perp,y)}\delta_{c,a}$ and  $V((a,b),c|(\perp,y),y)=\delta_{b,f(\perp,y)}\delta_{c,b}$.
    Thus, we may assume that upon receiving a query of the form $(\perp,y)$, Alice always outputs $(0,f(\perp,y))$ and upon receiving the query $\perp$, Bob always outputs $0$ (indeed, this will only increase their winning probability).
    Now, it is direct to see that any such quantum oracularisable strategy for $G^{\orac}$ corresponds to a quantum oracularisable correlation $C((x,y),z)=\tr (A_{x,y} B_z)$ for $(G_\xor)'$ where $B_\perp=I$ and $A_{\perp,y}=I$. Recall that $(G_{\xor})'$ is the XOR game corresponding to the oracularisation of $G_{\xor}$ constructed in \cref{lem:xor-orac}.
    
    To get the upper bound, consider such quantum oracularisable correlation $C((x,y),z)=\tr(A_{x,y}B_z)$.
    We have
    \begin{align*}
        \beta(G^{\orac},C)&=\sum_{x\in X\backslash\{\perp\}, y\in Y}\frac{1}{2}\pi(x,y)\parens*{\tr(A_{x,y}B_x+(-1)^{f(x,y)}A_{x,y}B_y)}\\
        &+\sum_{ y\in Y}\frac{1}{2}\pi(\perp,y)\parens*{\tr(I+(-1)^{f(x,y)}B_y)}\\
        &\leq \frac{1}{2}+\frac{1}{2}\sum_{x\in X, y\in Y}\pi(x,y)(-1)^{f(x,y)}\tr(B_xB_y)\\
        &=\frac{1}{2}+\frac{1}{2}\beta(G,C'),
    \end{align*}
    where $C'$ is the quantum correlation defined by $C'(x,y)=\tr(B_xB_y)$.
    The upper bound follows by considering the supremum over all quantum correlations.

    Finally, for the lower bound, fix a quantum oracularisable correlation for $G_\xor$, $C(x,y)=\tr(B_xB_y)$, with the convention that $B_\perp=I$.
    We define a quantum oracularisable correlation for $G_\xor'$, by setting $C'((x,y),z)=\tr(B_xB_z)$.
    Since $B_x$ and $B_y$ commute (whenever $\pi(x,y)>0$),  $C'$ is a well defined quantum oracularisable correlation.
    Then, the bias is
    \begin{align*}
        \beta(G^{\orac},C')&=\sum_{x\in X, y\in Y}\frac{1}{2}\pi(x,y)\parens*{\tr(A_{x,y}B_x+(-1)^{f(x,y)}A_{x,y}B_y)}\\
        &=\sum_{x\in X, y\in Y}\frac{1}{2}\pi(x,y)\parens*{\tr(I+(-1)^{f(x,y)}B_xB_y)}\\
        &=\frac{1}{2}+\frac{1}{2}\beta(G,C).
    \end{align*}

    The lower bound follows by considering the supremum over all possible quantum oracularisable correlations $C$.
\end{proof}

Let $1\geq c\geq s\geq 0.$ $\text{TXOR}_{c,s}^{\ast,\orac}$ is the problem of deciding, for a given tilted XOR game~$G$, if $\omegaqo(G)\geq c$ or $\omegaqo(G)<s$, with the promise that one of these two holds.
We denote by $\text{TXOR-MIP}^{\ast,\orac}_{c,s}$ the succinctly-presented version of this problem.

\begin{theorem}\label{thm:qo is hard}
    There exist $c>s$ such that there is a polynomial-time reduction from the halting problem to $\text{TXOR-MIP}^{\ast,\orac}_{c,s}$.
\end{theorem}

The reduction is the same as in \cref{thm:main-robust}, so we may take $c=\frac{3}{4}-\varepsilon$ for any $\varepsilon>0$ and $s=\frac{3}{4}-10^{-8}$.

\begin{proof}
    In \cite{TV25}, the authors construct a succinct E3-LIN protocol for the halting problem, while using the reduction of \cite{DFN+25}.
    As observed in \cite{CM24}, the latter reduction preserves oracularisability.
    In particular, the YES instances of their reduction admit a perfect quantum oracularisable strategies.
    Evidently, the completeness argument in \cite[Lemma 4.2]{TV25} also preserves oracularisability.

    Next, we show that the completeness argument of the present work, namely, \cref{lem:completeness}, can be adapted so that it will preserve oracularisability.
    Since for any nonlocal game we have $\omega^*(G)\geq \omegaqo(G)$, this will complete the proof. 

   Let $\{P^u_a\}_a$, $\{\Pi_\mbf{a}\}_\mbf{a}$ and $\ket{\psi}$ be as in \cref{lem:completeness}, and suppose $\ket{\psi}$ is the maximally entangled state, and that $[P^u_a,\Pi_\mbf{a}]=0$ for every $u, a$ and $\mbf{a}$.
   Here we are not changing $\Pi_\mbf{a}$, allowing
   $\Pi_\mbf{a}\neq0$ even if $a_1+a_2+a_3=1$.
   Define the operators $C_u$ similarly to their definition in \cref{lem:completeness}. Note that they commute with $A_{u'}=P^{u'}_0-P^{u'}_1$.
   The analysis of the success probability is similar to that in \cref{lem:completeness}, with the only difference in the last inequality, which here turns to $1-\varepsilon\leq\frac{1}{2}+\frac{4}{3}\beta(G^{\tiltedsquare},C)-\frac{1}{6}+2\frac{\varepsilon}{6}$,
    using the fact that $\sum_{a_1+a_2+a_3=1}\braket{\psi}{\Pi_\textbf{a}\otimes I}{\psi}\leq \varepsilon$.
    Therefore, $\beta(G^{\tiltedsquare},C)\geq\frac{1}{2}-\varepsilon.$
\end{proof}

\begin{corollary}
    There exist  constants $\frac{1}{2}<s'<c'<1$ such that for a given tilted XOR game $G$, it is $\tsf{RE}$-hard to decide whether $\omegaqo(G^{\orac})\leq s'$ or $\omegaqo(G^{\orac})\geq c'$.
\end{corollary}

\begin{proof}
    Let $c$ and $s$ be as in \cref{thm:main-robust}  , fix $x\in\{0,1\}^*$ and let $G_x$ be the resulted tilted XOR game of our reduction for the input $x$.
    If $x$ is a no instance, then it is promised that $\omegaqo(G_x)\leq\omega^*(G_x)\leq s$.
    Thus, by \cref{lem:xor-orac-2}, we have
    $\omegaqo(G_x^{\orac})\leq\frac{1+\omega^*(G_x)}{2}\leq \frac{1+s}{2}.$

    If, instead, $x$ is a yes instance (of the halting problem), then it follows the proof of \cref{thm:qo is hard}, that $\omegaqo(G_x)\geq c$.
    Again, by \cref{lem:xor-orac-2}, we get
    $\frac{1+c}{2}\leq \frac{1+\omegaqo(G_x)}{2}\leq \omegaqo(G^{\orac}).$

    Thus, we can take $s'=\frac{1+s}{2}$ and $c'=\frac{1+c}{2}$.
\end{proof}

We finish this section with a consequence for the hardness of XOR games under oracularisable strategies.

\begin{corollary}
    There exist $c>s$ such that there is a polynomial-time reduction from the halting problem to $\text{XOR-MIP}^{\ast,\orac}_{c,s}$.
\end{corollary}

We can take the same values of $c$ and $s$ as in \cref{thm:qo is hard}. This proves a weaker version of a consequence of the quantum unique games conjecture shown in~\cite[Theorem 26]{MS25}. In our language, they show that, assuming their version of the quantum unique games conjecture, for every $t\in(\frac{1}{2},1)$ and sufficiently small $\varepsilon>0$, $\text{XOR}^{\ast,\orac}_{1-O(\varepsilon),1-O(\varepsilon^t)}$ is $\tsf{RE}$-complete.

\begin{proof}
    Let $x$ be an instance of the halting problem and let $G_x$ be the corresponding tilted XOR game constructed via the reduction of \cref{thm:qo is hard}. We will show that $\omegaqo(G_x)=\omegaqo\bigl((G_x)_{\xor}\bigr)$, getting $\tsf{RE}$-hardness for the quantum oracularisable value of XOR games. First, note that the distinguished question $\perp$ in the tilted bipartite cube game is asked to Alice with every one of Bob's questions. Hence, in $(G_x)_{\xor}$, Alice's observable for $\perp$ commutes with every one of Bob's observables in any quantum oracularisable strategy. Fix a strategy for $(G_x)_{\xor}$ and let $\Pi$ be the projector onto the $1$ eigenspace of $A_\perp$. Consider the new strategy with $A_x$ replaced with $\Pi A_x\Pi+(I-\Pi)A_x(1-\Pi)$. This has the same value as the original strategy since $\Pi$ commutes with all the $B_y$, and it has the property that $A_\perp$ commutes with all the observables. Hence, we can multiply all the observables by $A_\perp$ to get a strategy where $A_\perp=I$ with the same value, as in the classical case. As such, $\omegaqo(G_x)=\omegaqo\bigl((G_x)_{\xor}\bigr)$. Thus, we get the wanted hardness.
\end{proof}

\subsection{Tracial strategies for tilted XOR games and noncommutative Max-Cut}\label{sec:tracial-strategies}

In this section, we show hardness of approximation of the value of tilted XOR games under quantum tracial strategies. These are strategies where the shared state is maximally entangled, and hence can be represented by a trace. We connect this hardness to the noncommutative Max-Cut problem, in the sense of~\cite{CMS2024}.

\begin{definition}
    Let $G$ be a nonlocal game. A strategy $p$ for $G$ is called \textbf{quantum tracial} if there exists a finite dimensional Hilbert space $H$ of some dimension $d$, collections of PVMs $\{P^x_a\}_a$ and $\{Q^y_b\}_b$ acting on $H$ such that $p(a,b|x,y)=\tr(P^x_aQ^y_b)$, where $\tr$ is the normalised trace.
The \textbf{quantum tracial value} of $G$ is the supremum over the values of all quantum tracial strategies; it is denoted $\omegatr(G)$.
\end{definition}

As previously, quantum tracial strategies give rise to a decision problem. Let $1\geq c\geq s\geq 0.$ $\text{TXOR}_{c,s}^{\ast,\tr}$ is the problem of deciding, for a given tilted XOR game~$G$, if $\omegatr(G)\geq c$ or $\omegatr(G)<s$, with the promise that one of these two holds.
We denote by $\text{TXOR-MIP}^{\ast,\tr}_{c,s}$ the succinctly-presented version of this problem.

\begin{theorem}\label{thm:tracial-txor}
    There exist $c>s$ such that there is a polynomial-time reduction from the halting problem to $\text{TXOR-MIP}^{\ast,\tr}_{c,s}$.
\end{theorem}

We again make use of the same reduction as in \cref{thm:main-robust}, so we have that $c=\frac{3}{4}-\varepsilon$ for any~$\varepsilon>0$ and $s=\frac{3}{4}-10^{-8}$.

\begin{proof}
    Given an instance of the halting problem $x$, let $H_x$ be the E3-LIN instance constructed in~\cite{TV25}, and let $G_x$ be the (succinctly-presented) tilted XOR game constructed in \cref{thm:main-robust}. Let $c_0=\frac{3}{4}-\varepsilon$ and $s_0$ be the constants from \cref{thm:main-robust}. First, we show soundness. If $x$ is a no instance of the halting problem, then $\omega^{\ast}(G_x)< s_0$. Since every quantum tracial strategy is a quantum strategy, $\omegatr(G_x)<s_0$ as well. To finish, we show completeness. Suppose $x$ is a yes instance. Then, $\omega^\ast(H_x)\geq 1-O(\varepsilon)$. But, using~\cite{Vid22,Cul26}, we know that near-perfect strategies for projection games may be rounded to tracial strategies. This gives that $\omegatr(H_x)\geq 1-O(\varepsilon^{1/4})$. Due to \cref{lem:completeness}, near-perfect quantum strategies for $H_x$ give rise to near-optimal quantum strategies for $G_x$, with the same shared state. In the current setting, this implies that that there is a near-optimal quantum strategy for $G_x$ with a maximally-entangled state, and hence a near-optimal tracial strategy. Formally, we find $\omegatr(G_x)\geq\frac{3}{4}-O(\varepsilon^{1/4})$.
    Taking $c=\frac{3}{4}-O(\varepsilon^{1/4})$ and $s=s_0$ gives the result.
    
\end{proof}

Note that the completeness bound in the above lemma can also be obtained from \cref{thm:qo is hard}, which shows that $\omegaqo(G_x)\geq c_0$, and since $\omegatr(G_x)\geq\omegaqo(G_x)$, the desired completeness bound with respect to tracial strategies follows.

This immediately implies hardness results for a variant of noncommutative Max-Cut.

\begin{corollary}[noncommutative Max-Cut]
    There exist $c>s$ such that it is $\tsf{RE}$-hard to decide if the objective value of the following optimisation is $\geq c$ or $<s$: given a graph $G$, maximise
    \begin{align}
       \frac{1}{|E(G)|}\sum_{(i,j) \in E(G)} \frac{1 - \tr(X_iX_j)}{2},
    \end{align}
    where each $X_i$ is hermitian and satisfies $X_i^2=I$, and $X_1=I$, and $\tr$ is the normalised trace.
\end{corollary}

\begin{proof}
    Let $c_0=\frac{3}{4}-\varepsilon>s_0=\frac{3}{4}-\delta$ be the constants from \cref{thm:tracial-txor}. Due to the theorem, it is $\tsf{RE}$-hard to decide if the objective value of the following optimisation problem is $\geq c_0$ or $<s_0$: given a probability distribution $\pi$ on $[n]^2$ and a function $f:[n]^2\rightarrow\{0,1\}$, maximise
    \begin{align*}
        \sum_{i,j\in[n]}\pi(i,j)\frac{1+(-1)^{f(i,j)}\tr(X_iX_j)}{2},
    \end{align*}
    where each $X_i$ is a hermitian unitary and $X_1=I$. This consists of equality and inequality relations; the main part of the reduction to Max-Cut consists of removing the equality relations. Fix an instance of this optimisation problem. We may assume without loss of generality that this instance is constructed from an E3-LIN instance via the construction used in this paper. Due to the construction based on the long code used in~\cite{TV25}, the probabilities of sampling an equation of the form $v_i\oplus v_j\oplus v_k=0$ and of the form $v_i\oplus v_j\oplus v_k=1$ are equal in the E3-LIN instance. This implies that the probability of sampling an equality relation is exactly $\frac{5}{16}$ (the two types of cubes are sampled with equal probability; see \cref{fig:tilted-cube-game-reduction}). Define two probability distribution $\pi_0$ and $\pi_1$ on $[n]^2$ as $\pi_0(i,j)=\frac{16}{5}\pi(i,j)$ if $f(i,j)=0$ and $\pi_0(i,j)=0$ otherwise; and $\pi_1(i,j)=\frac{16}{11}\pi(i,j)$ if $f(i,j)=1$ and $\pi_1(i,j)=0$ otherwise. By the above they are probability distributions, and they are disjoint by construction. Fix $\Delta\in(0,1)$. Now, consider the following optimisation problem: maximise
    \begin{align*}
        \sum_{i,j}\parens*{\frac{11}{16}\Delta\pi_1(i,j)\frac{1-\tr(X_iX_j)}{2}+\frac{5}{16}\Delta\pi_0(i,j)\frac{1-\tr(Y_iX_j)}{2}+(1-\Delta)\pi_0(i,j)\frac{1-\tr(Y_iX_i)}{2}},
    \end{align*}
    where the $X_i$ and $Y_i$ are hermitian unitaries, and $X_1=I$. To show completeness, if the original instance had objective value $\geq\frac{3}{4}-\varepsilon$, then we can take $Y_{i}=-X_i$, and see that the new instance has objective value $\geq\Delta\parens*{\frac{3}{4}-\varepsilon}+(1-\Delta)=1-\frac{1}{4}\Delta-\Delta\varepsilon\eqqcolon c$. To show soundness, suppose that there exists some $\eta>0$ such that the objective value of the new instance is $\geq1-\eta\eqqcolon s$. Then, we have in particular that
    \begin{align*}
        \sum_{i,j}\pi_0(i,j)\frac{1-\tr(Y_iX_i)}{2}\geq\frac{1-\eta-\Delta}{1-\Delta},
    \end{align*}
    and hence $\sum_{i,j}\pi_0(i,j)\norm{X_i+Y_i}_{f}^2\leq4\frac{\eta}{1-\Delta}$. This implies that
    \begin{align*}
        \sum_{i,j}&\pi(i,j)\frac{1+(-1)^{f(i,j)}\tr(X_iX_j)}{2}\\
        &\geq\sum_{i,j}\parens*{\frac{11}{16}\pi_1(i,j)\frac{1-\tr(X_iX_j)}{2}+\frac{5}{16}\pi_0(i,j)\frac{1-\tr(Y_iX_j)}{2}-\frac{5}{32}\pi_0(i,j)\norm{X_i+Y_i}_f}\\
        &\geq\frac{\Delta-\eta}{\Delta}-\frac{5}{16}\sqrt{\frac{\eta}{1-\Delta}}.
    \end{align*}
    Now take $\Delta=\frac{\delta^2}{2}$ and $\eta=(\frac{1}{4}+\frac{\delta}{2})\frac{\delta^2}{2}$. First, we have that $s=1-\eta=1-\frac{\Delta}{4}-\frac{\delta\Delta}{2}<1-\frac{1}{4}\Delta-\Delta\varepsilon=c$ for $\varepsilon$ small enough, giving positive completeness-soundness gap. Also,
    \begin{align*}
        \frac{\Delta-\eta}{\Delta}-\frac{5}{16}\sqrt{\frac{\eta}{1-\Delta}}&=1-\parens*{\frac{1}{4}+\frac{\delta}{2}}-\frac{5}{16}\sqrt{\frac{(\frac{1}{4}+\frac{\delta}{2})\frac{\delta^2}{2}}{1-\frac{\delta^2}{2}}}\\
        &\geq\frac{3}{4}-\frac{\delta}{2}-\frac{5}{16}\frac{\delta}{2}>\frac{3}{4}-\delta=s_0,
    \end{align*}
    finishing the proof of soundness of the reduction. To finish, note that we can get rid of the probability distribution by duplicating edges.
\end{proof}

To finish, note that there is a connection to noncommutative polynomial optimisation, as studied in~\cite{PNA10,BKP16,KMV22}. The results imply that optimising the normalised trace of a quadratic noncommutative polynomial over hermitian unitaries is $\tsf{RE}$-hard; this follows due to~\cite{JNV+20arxiv,MNY22}, but with more constraints and a more complicated form of the objective function.

 \section*{Acknowledgments}

RC received support from an NSERC Alliance grant under the project QUORUM. EC was supported in part by an NSERC CGS D. We thank Laura Man\v{c}inska, for discussions about the open question of hardness of binary games and for informing us of~\cite{Rus20}; William Slofstra, for general discussions and for providing references for the exact feasibility of SDPs; and Thomas Vidick, for discussions about the connection to unique games.

\newpage

\bibliographystyle{bibtex/bst/alphaarxiv.bst}
\bibliography{bibtex/bib/full.bib,bibtex/bib/quantum.bib,bibtex/quantum_new.bib}

\end{document}